\documentclass[12pt]{article}
\usepackage{url}
\usepackage{amsfonts}
\usepackage{amsmath,amssymb,color}
\usepackage{graphicx}
\usepackage{enumitem}
\usepackage{multirow}
\usepackage{bm}
\usepackage{natbib}
\usepackage{verbatim}
\usepackage{float}
\usepackage[toc,page]{appendix}
\usepackage[onehalfspacing]{setspace}
\usepackage{setspace}
\usepackage{adjustbox}
\usepackage{rotating}
\usepackage{subcaption}
\usepackage{amsthm}
\usepackage{booktabs,dcolumn,caption}

\usepackage{array}
\newcolumntype{L}[1]{>{\raggedright\let\newline\\\arraybackslash\hspace{0pt}}p{#1}}
\newcolumntype{C}[1]{>{\centering\let\newline\\\arraybackslash\hspace{0pt}}p{#1}}
\newcolumntype{R}[1]{>{\raggedleft\let\newline\\\arraybackslash\hspace{0pt}}p{#1}}

\newcommand{\E}{\mathbf 1}

\newcommand{\EEE}{\mathbb E}

\newcommand{\0}{\mathbf 0}

\newcommand{\ttt}{\boldsymbol \theta}

\newcommand{\aaaa}{\mathbf a}

\newcommand{\dd}{\mbox{$\mathbf d$}}

\newcommand{\R}{\mathbb R}

\newcommand{\uu}{\mathbf u}

\newcommand{\vv}{\mathbf v}

\newcommand{\xx}{\mathbf x}

\newcommand{\1}{\uppercase\expandafter{\romannumeral1}}
\newcommand{\2}{\uppercase\expandafter{\romannumeral2}}

\newcommand{\argmax}{\operatornamewithlimits{arg\,max}}

\newtheorem{theorem}{Theorem}

\newtheorem{algorithm}{Algorithm}

\newtheorem{lemma}{Lemma}

\newtheorem{proposition}{Proposition}
\newtheorem{remark}{Remark}

\title{A Note on Exploratory Item Factor Analysis by Singular Value Decomposition}
\author{Haoran Zhang, Yunxiao Chen and Xiaoou Li}
\date{}

\begin{document}
\maketitle

\doublespacing

\begin{abstract}
We revisit a singular value decomposition (SVD) algorithm given in \cite{chen2019joint} for exploratory Item Factor Analysis (IFA). This algorithm estimates a multidimensional IFA model by SVD and was used to obtain a starting point for joint maximum likelihood estimation in \cite{chen2019joint}.
Thanks to the analytic and computational properties of SVD, this algorithm guarantees a unique solution
and has computational advantage over other exploratory IFA methods.
Its computational advantage becomes significant when the numbers of respondents, items, and factors are all large.
This algorithm can be viewed as a generalization of principal component analysis
(PCA) to binary data.
In this note, we provide the statistical underpinning of the algorithm. In particular, we show its statistical consistency
under the same double asymptotic setting as in \cite{chen2019joint}. We also demonstrate how this algorithm provides a scree plot for investigating the number of factors and provide its asymptotic theory.
Further extensions of the algorithm are discussed.  Finally, simulation studies suggest that
the algorithm has good finite sample performance.


\end{abstract}	


\section{Background}\label{Sec:intro}

Exploratory IFA \citep{bock1988full} has
been widely used for analyzing item-level data in social and behavioral sciences \citep{bartholomew2008analysis}.
We consider a standard exploratory IFA setting for binary item response data. Let $Y_{ij} \in \{0, 1\}$ be a random variable, denoting individual $i$'s response to item $j$, where $i = 1, ..., N$, and $j = 1, ..., J$.
Moreover, IFA assumes that  an individual $i$'s responses are driven by $K$ latent factors, denoted by $\ttt_i = (\theta_{i1}, ..., \theta_{iK})^\top$. We consider a general family of multidimensional IFA models \citep{reckase2009multidimensional}, which assumes that
\begin{equation}\label{eq:IRF}
\Pr(Y_{ij} = 1 \vert \ttt_i) = f(d_{j} + \aaaa_j^\top \ttt_i),
\end{equation}
where $\aaaa_j = (a_{j1}, ..., a_{jK})^\top$ is typically known as the loading parameters, $d_j$ is an intercept parameter,  and $f: \mathbb R \mapsto (0, 1)$ is a pre-specified monotone increasing function which guarantees \eqref{eq:IRF} to be a valid probability.
Using the terminology from generalized linear models, $f$ is called the inverse link function.
Note that \eqref{eq:IRF} includes the widely used multidimensional two-parameter logistic (M2PL) model and multidimensional  normal ogive model as special cases, for which  $f(x) = \exp(x)/(1+\exp(x))$  and $f(x) = \int_{-\infty}^x \exp(-t^2/2)/(2\pi)dt$, respectively. Moreover, we assume local independence; that is, $Y_{i1}$, ..., $Y_{iJ}$ are conditionally independent given $\ttt_i$. Finally, $\ttt_i$, $i =1, ..., N$, are independent and identically distributed, following an unknown distribution $F$.



A major focus of exploratory IFA is to estimate the loading matrix $A = (a_{jk})_{J\times K}$,
which helps to understand the latent structure underlying the set of items.
It is worth noting that the loading matrix can only be recovered up to an oblique rotation \citep{browne2001overview}.\footnote{We only discuss oblique rotation here, as our general exploratory IFA model does not require the factors to be uncorrelated. If the factors are further required to be uncorrelated, then the loading matrix can be recovered up to an orthogonal rotation, for which the rotation matrix $O$ is an orthogonal matrix \citep[e.g.,][]{kaiser1958varimax}.} That is, model~\eqref{eq:IRF} will remain unchanged, with a rotated loading vector $\tilde \aaaa_j = O^\top \aaaa_j$ and $\tilde \ttt_i = O^{-1}\ttt_i$, where $O$ is an $K\times K$ invertible matrix that is also known as an oblique rotation.
Recognizing the rotational indeterminacy issue, exploratory IFA typically proceeds in two steps.
In the first step, an estimate $\hat A$ is obtained, using an arbitrary way to fix the rotation.
Then in the second step,
analytic rotational methods are applied to $\hat A$ to obtain a more sparse loading matrix for better interpretability.
An analytic rotation finds a rotation matrix $O$ such that $\hat A O$ minimizes a certain ``complexity function", where a lower value of the complexity function indicates more sparsity in the loading matrix \citep[see][for a review of analytic rotations]{browne2001overview}.
It implicitly assumes that the true loading matrix has a sparse pattern; i.e., each item is only directly associated with a small number of factors.

In this note, we focus on the first step of exploratory IFA. In particular, we study an estimator given in \cite{chen2019joint} that is based on
SVD. Comparing with other estimators, this estimator is computationally much faster and does not suffer from convergence issues.
It was used to obtain a starting point for a constrained joint maximum likelihood estimator (CJMLE). Simulation studies showed that the convergence of CJMLE can be improved by using the SVD-based estimator as a starting point. Moreover, this SVD-based estimator itself is reasonably accurate when both $N$ and $J$ are large. Thus, it can be used not only  as a starting point for the CJMLE, but also as a quick and high-quality solution to large-scale exploratory IFA problems. In what follows, we investigate the statistical properties of this estimator.




\section{Main results}

\paragraph{SVD-based estimator.} We restate this SVD-based algorithm below.\footnote{The original algorithm was described in the supplementary material of \cite{chen2019joint}.
The algorithm here is a slightly modified version. The major modification is in step 3 of the algorithm that requires at least $K+1$ singular values to be retained. This modification can improve the finite-sample performance of the algorithm; see Remark~\ref{rmk:mod} for more discussions. The other modifications are mainly to simplify the exposition of the algorithm.}

\begin{algorithm}[SVD-based estimator for exploratory IFA]\label{alg:SVD}
\
\begin{enumerate}
\item Input response $Y = (y_{ij})_{N\times J}$, the number of factors $K$, inverse link function $f$, and truncation parameter $\epsilon_{N,J} > 0$.
\item Apply the singular value decomposition to $Y$ and obtain $Y = \sum_{j = 1}^J \sigma_j \uu_j\vv_j^\top$, where $\sigma_1 \geq ...\geq \sigma_J \geq 0$ are the singular values, and $\uu_j$s and $\vv_j$s are left and right singular vectors, respectively.
\item Let $X = (x_{ij})_{N \times J} = \sum_{k = 1}^{\tilde K} \sigma_k \uu_k\vv_k^\top,$ where
    $\tilde K = \max\big\{K+1, \argmax_k\{\sigma_k \geq 1.01 \sqrt{N}\}\big\}$.
\item Let $\hat X = (\hat x_{ij})_{N\times J}$ be defined as
$$
\hat x_{ij} =
\begin{cases}
\epsilon_{N,J}, \quad \text{if } x_{ij} < \epsilon_{N,J},\\
x_{ij}, \quad \text{if } \epsilon_{N,J} \leq x_{ij} \leq 1-\epsilon_{N,J},\\
1-\epsilon_{N,J}, \quad \text{if } x_{ij} > 1 - \epsilon_{N,J}.
\end{cases}
$$
\item Let $\tilde M = (\tilde m_{ij})_{N\times J},$ where $\tilde m_{ij} = f^{-1}(\hat x_{ij}).$
\item Let $\hat\dd = (\hat d_1,...,\hat d_J)$, where $\hat d_j = (\sum_{i=1}^N\tilde m_{ij})/N$.
\item Apply singular value decomposition to $\hat M = (\tilde m_{ij} - \hat d_j)_{N\times J}$ to have $\hat M = \sum_{j = 1}^J \hat\sigma_j \hat\uu_j\hat\vv_j^\top$, where $\hat\sigma_1 \geq ...\geq \hat\sigma_J \geq 0$ are the singular values, and $\hat\uu_j$s and $\hat\vv_j$s are the left and right singular vectors, respectively.
\item Output $\hat A = \frac{1}{\sqrt{N}}(\hat \sigma_1 \hat\vv_1,...,\hat\sigma_K\hat\vv_K), \hat \Theta = \sqrt{N}(\hat\uu_1,...,\hat\uu_K).$
\end{enumerate}
\end{algorithm}

\begin{remark}
SVD is a powerful tool for the factorization of rectangular matrices that has been widely used in multivariate statistics for the dimension reduction of data \citep{wall2003singular}.
Thanks to the mathematical properties of SVD, the estimator given by Algorithm~\ref{alg:SVD} is analytic that does not suffer from convergence issues. On the other hand, as the objective functions of
the CJMLE and the marginal maximum likelihood estimator \citep[MMLE;][]{bock1981marginal} are nonconvex, there is no guarantee for finding their global optima.
In addition, this SVD approach is also much faster than the other estimators, including the CJMLE and MMLE. In particular, the computation of the MMLE based on the vanilla expectation maximization algorithm is not affordable when the latent dimension $K$ is of a moderate size (e.g., $K\geq 5$). Even the stochastic algorithms for the MMLE \citep{cai2010high,cai2010metropolis,zhang2018improved}  and the alternating minimization algorithm for the CJMLE \citep{chen2019joint,chen2017structured} are much slower than the SVD algorithm,
as these algorithms typically need a large number of iterations to converge. A speed comparison is provided in the simulation study between the SVD method and the CJMLE.


\end{remark}

\begin{remark}\label{rmk:PCA}
Algorithm~\ref{alg:SVD} can be viewed as a generalization of PCA to binary data.
PCA is an SVD-based algorithm \citep[e.g., Chapter 14,][]{friedman2001elements} that is fast and commonly used for exploratory linear factor analysis.
Unfortunately, PCA cannot be applied to exploratory IFA, due to the nonlinear link function in IFA models.
Unlike PCA which applies SVD only once, Algorithm~\ref{alg:SVD} applies SVD twice. The first application of SVD and the inverse transformation (steps 2-5) denoise and linearize the data. Then, the second application of SVD (steps 6-7) is essentially doing PCA to the linearized data.


\end{remark}

\begin{remark}

Similar as the  CJMLE \citep{chen2019joint,chen2017structured}, this SVD-based estimator does not require the latent distribution $F$ to be known or to take a parametric form as is required in the MMLE approach.
Moreover, exploratory IFA based on
tetrachoric/polychoric correlations \citep{muthen1984general,lee1990full,lee1992structural,joreskog1994estimation}
or composite-likelihood-based estimator \citep{katsikatsou2012pairwise}
requires $F$ to be multivariate normal, with the former approach further requiring the inverse link
$f$ to be probit. In this sense, the SVD-based estimator and the CJMLE require less model assumptions than the other estimators. As a price, their consistency requires stronger conditions, specifically, a double asymptotic regime where both $N$ and $J$ diverge.

\end{remark}

\begin{remark}\label{rmk:mod}
  Steps 2-4 of the algorithm essentially follow the same procedure of \cite{chatterjee2015matrix} for matrix estimation. We thus refer the readers to \cite{chatterjee2015matrix} for the details. A small difference is that we require $\tilde K \geq K+1$ in step 3 of the algorithm. This modification does not affect the asymptotic behavior of the estimator.
  However, it can improve the finite-sample performance when $N$ and $J$ are not large enough. Intuitively, we need $\tilde K$ to be at least $K+1$, in order to recover the matrix $(d_{j} + \aaaa_j^\top \ttt_i)_{N\times J}$ which is of rank $K + 1$. The constant  1.01 in step 3 of the algorithm follows Theorem 1.1 of \cite{chatterjee2015matrix}, which makes use of the fact that $Var(Y_{ij})\leq 1/4$.
  This constant can be replaced by any fixed constant in the open interval $(1, 1.5)$, without affecting its consistency given in Theorem~\ref{thm:consis} below. We set it to be 1.01, because
  according to Theorem 1.1 of \cite{chatterjee2015matrix} this constant should be chosen close to 1 for better accuracy.
\end{remark}

\begin{remark}
The truncation step (step 4) is necessary, as it guarantees the existence of a solution. This is because, even though $x_{ij}$ in step 3 is approximating the true probability $\Pr(Y_{ij} = 1)$, it is
not guaranteed to be in the interval $(0, 1)$. As a consequence, $f^{-1}(x_{ij})$ may not be well-defined. The pre-specified truncation parameter $\epsilon_{N,J} > 0$ determines the truncation level. As shown in the sequel, the choice of $\epsilon_{N,J}$ affects the statistical consistency of the proposed algorithm. Under certain circumstances, we will need the truncation parameter $\epsilon_{N,J}$ to decay to zero as $N$ and $J$ grow to infinity, which is why we attach subscripts $N$ and $J$ to the truncation parameter.
In practice, the performance of the proposed method tends to be insensitive to the choice of $\epsilon_{N,J}$ when it is chosen sufficiently small, which is justified theoretically by Propositions~\ref{prop:bounded} and \ref{prop:normal} below, under two specific settings.  In the numerical analysis of this paper, we use $\epsilon_{N,J} = 10^{-4}$ as a default value.
\end{remark}

\paragraph{Statistical consistency.}
In what follows, we establish the theoretical consistency of this method. In particular, we
show that this SVD-based algorithm is consistent under similar asymptotic setting and notion of consistency
as in \cite{chen2019joint} and \cite{chen2017structured}. 
The proofs of our theoretical results are given in the supplementary material.
More precisely, we consider a loss function on the recovery of the true loading matrix $A^* = (a_{jk}^*)_{J\times K}$ up to an oblique rotation
\begin{equation}\label{eq:ave}
L_{N, J}(A^*, \hat A) =\min_{O\in \mathbb R^{K\times K}} \left\{ \frac{\Vert A^* - \hat A O \Vert_F^2}{JK} \right\},
\end{equation}
where the subscripts $N$ and $J$ are used to emphasize that the loss function depends on the sample size $N$ and the number of items $J$, and $\Vert X\Vert_F = \sqrt{\sum_{i}\sum_{j} x_{ij}^2}$ denotes the Frobenius norm of a matrix $X = (x_{ij})$.
Under mild technical conditions and a double asymptotic setting where both $N$ and $J$ grow to infinity, we show the loss function $L_{N, J}(A^*, \hat A) $ converges to zero in probability.
The regularity conditions and the consistency result are formally described in Theorem~\ref{thm:consis}, with two special cases discussed in the sequel.
Similar double asymptotic settings
have been considered in psychometric research, including the analyses of unidimensional IRT models \citep{haberman1977maximum,haberman2004joint} and diagnostic
classification models \citep{chiu2016joint}.
The following regularity conditions are needed for our main result in Theorem~\ref{thm:consis}. As will be discussed in the sequel, these conditions are mild.
\begin{itemize}
\item[A1.] There exists a constant $C$ such that $\sqrt{(d_j^*)^2 + \|\aaaa^*_j\|^2 } \leq C$, for $j = 1,...,J$, where $d_j^*$ and $\aaaa_j^*$ are the true item parameters.

\item[A2.] The true person parameters
$\ttt_1^*,...,\ttt_N^*$ are independent and identically distributed (i.i.d.) following a distribution $F$ which has mean $\0$ and positive definite covariance matrix $\Sigma.$ 

\item[A3.] The inverse link function $f$ is strictly monotone increasing, continuously differentiable, and Lipschitz continuous with Lipschitz constant $L$. We further assume that $$\lim_{x\rightarrow -\infty} f(x) = 0, ~~\mbox{and}~~ \lim_{x\rightarrow \infty} f(x) = 1.$$


\item[A4.] There exists a constant $C_1,$ such that the $K$th singular value of $A^*$, denoted by $\sigma_K(A^*)$, satisfies $\sigma_K(A^*) \geq C_1\sqrt{J}$ for all $J$.

\item[A5.] The sample size $N$ is no less than the number of items $J$, i.e., $N \geq J$.
\end{itemize}

\begin{theorem}\label{thm:consis}
Suppose that conditions A1-A5 are satisfied. Further suppose that $\epsilon_{N,J} \leq \frac{1}{5}$ and satisfies
\begin{align}
& \Pr\left( \|\ttt_1^*\| \geq h(2\epsilon_{N,J})/C  \right)= o({N}^{-1}),\label{eq:epsilon1}\\
&\frac{(h(2\epsilon_{N,J}))^{\frac{K+1}{K+3}}}{(\epsilon_{N,J} g(\epsilon_{N,J}))^2} = o(J^{\frac{1}{K+3}}),\label{eq:epsilon2} 
\end{align}
where
\begin{align}
h(y) &= \max\{ |f^{-1}(y)|, |f^{-1}(1-y)| \}, ~~ y \in (0,0.5),\\
g(y) &= \inf \{f'(x): x \in [f^{-1}(y),f^{-1}(1-y)]\}, ~~ y \in (0, 0.5).
\end{align}
Then the estimate $\hat A$ given by Algorithm~\ref{alg:SVD} satisfies
 $L_{N,J}(A^*,\hat A) \overset{pr}{\to} 0$, as $N, J \rightarrow \infty.$
\end{theorem}

\begin{remark}
We remark that the notion of consistency for the estimation of the loading matrix is weaker than that in the traditional sense, since the loss function \eqref{eq:ave} is an average of the entry-wise losses when $J$ grows. Let $\tilde O$ minimize the right hand side of \eqref{eq:ave} and let $\tilde A := (\tilde a_{jk})_{J\times K} =  \hat A \tilde O$. Then \eqref{eq:ave} converges to 0 means that for any $\epsilon > 0$,
$({\sum_{j=1}^J\sum_{k=1}^K 1_{\{\vert a_{jk}^*  - \tilde a_{jk} \vert > \epsilon\}}})/{JK}$ also converges to 0.
That is, the proportion of inaccurately estimated loading parameters converges to zero in probability under the optimal rotation. Due to the double asymptotic setting, our theoretical result only suggests the sensible use of the SVD-based algorithm when the sample size $N$ and the number of items $J$ are both large.
\end{remark}

\begin{remark}
It has been well-understood that PCA can consistently estimate a linear factor model under a similar double asymptotic setting \citep{stock2002forecasting}, which provides the theoretical justification for the use of PCA in exploratory linear factor analysis.
Theorem~\ref{thm:consis} can be viewed as a similar result for exploratory item factor analysis.
\end{remark}

\begin{remark}
We provide some discussions on the regularity conditions required in Theorem~\ref{thm:consis}. Assumption~A1 requires that the parameters of each item, including the intercept and slope parameters, should not be too large. That is, the presence of an extreme item is likely to distort the analysis. Assumption A2 is a very standard assumption in exploratory IFA. It is more flexible than many exploratory IFA settings,
 as it does not require the distribution $F$ to be multivariate normal.
 Assumption A3 is satisfied by the logistic and probit link functions,  two most commonly used link functions in exploratory IFA, but it excludes, for example, the multidimensional version of the three-parameter logistic model, as a special case. Assumption A4 requires that there is sufficient variability in the items. The same assumption is also required in \cite{chen2019joint} and \cite{chen2017structured}. In fact, this assumption is satisfied with probability tending to one, when the true loadings $\aaaa_j^*$ are i.i.d. samples from a $K$-variate distribution whose covariance matrix is non-degenerate. Finally, assumption A5 is
 practically reasonable, as in large-scale measurement, the sample size is usually larger than the number of items. Since people and items are almost mathematically symmetric
 in the IFA model, similar asymptotic results can be derived when $J\geq N$.

\end{remark}

\begin{remark}
We further provide some intuitions on the reason why the algorithm works.
Steps 2-4 essentially follow the same procedure of \cite{chatterjee2015matrix} for matrix estimation.
The procedure guarantees the loss ${\sum_{i,j}( f(d_j^* + (\aaaa_j^*)^\top \ttt_i^*) - \hat x_{ij} )^2}/{(NJ)}$
to be small with high probability, where $d_j^*$ and $\aaaa_j^*$ denote the true item specific parameters and $\ttt_i^*$ denotes the true person parameters sampled from distribution $F$.
Further with conditions A1 and A3, steps 5 and 6 guarantee the average loss
$\sum_{i=1}^N \sum_{j=1}^J ((\aaaa_j^*)^\top \ttt_i^* - \hat \aaaa_j^\top \hat \ttt_i)^2/(NJ)$
to be small with high probability. Finally, under conditions A2 and A4,
the famous Davis-Kahan-Wedin theorem from matrix perturbation theory \citep[see e.g.,][]{stewart1990matrix,o2018random} guarantees that $L_{N, J}(A^*, \hat A)$ is small with high probability.
\end{remark}

\begin{remark}\label{rmk:eps}
Equations~\eqref{eq:epsilon1} and \eqref{eq:epsilon2} are requirements on the truncation parameter $\epsilon_{N,J}$, which depends on both the tail of distribution $F$ and the properties of the inverse link function. Roughly speaking, Equation~\eqref{eq:epsilon1} is saying that
$\epsilon_{N,J}$ cannot be too large. This is because, given $F$ and $f$, the probability in \eqref{eq:epsilon1} is increasing in $\epsilon_{N, J}$. Requiring the probability being $o(N^{-1})$ implies that $\epsilon_{N, J}$ cannot be large.
This requirement is intuitive, because $\tilde M$ can be a poor approximation to $M^* = (m_{ij}^*)_{N\times J} := (d_{j}^* + (\aaaa_j^*)^\top \ttt_i^*)_{N\times J}$, when many entries of $M^*$ are larger than $h(\epsilon_{N, J})$. The function $h(\cdot)$ transforms the truncation on $x_{ij}$ to a truncation on $\tilde m_{ij}$. Using $h(2\epsilon_{N, J})$ instead of $h(\epsilon_{N, J})$ is for technical reasons.

Equation~\eqref{eq:epsilon2} requires that $\epsilon_{N,J}$ cannot be too small, as the left hand side of \eqref{eq:epsilon2} is decreasing in $\epsilon_{N,J}$. This requirement is also intuitive. 
Note that  $|\tilde m_{ij}| \leq h(\epsilon_{N, J})$, where $h(\epsilon_{N, J})$ is decreasing in $\epsilon_{N,J}$.
Therefore, a sufficiently large choice of $\epsilon_{N,J}$
avoids the approximation error $\Vert \tilde M - M^*\Vert_F$ being too large when there exist
some extreme estimates $\tilde m_{ij}$.
Function $g(\cdot)$ measures the local flatness of the inverse link $f$. The true matrix $M^*$ is more difficult to estimate when
$g(\epsilon_{N,J})$ is smaller. This is because $|\tilde m_{ij} - m_{ij}^*|$ can be large, even when $|\hat x_{ij} - f(m_{ij}^*)|$ is small, due to the local flatness of the inverse link function.

\end{remark}

\begin{remark}\label{rmk:random design}

We take a stochastic design for the true person parameters and a fixed design for the true item parameters,
following the convention of item factor analysis \citep[e.g.,][]{bartholomew2008analysis}.
It is worth pointing out that whether taking a stochastic or fixed design is not essential under our double asymptotic regime.
For example, the consistent result of Theorem \ref{thm:consis} still holds, if we can replace condition A2 by a corresponding fixed design as in  \cite{chen2019joint}.
\end{remark}

Following the discussion on $\epsilon_{N,J}$ in Remark~\ref{rmk:eps}, we consider two concrete settings under which the requirement on
$\epsilon_{N,J}$ becomes more specific. These results are given in Propositions~\ref{prop:bounded} and \ref{prop:normal}.

\begin{proposition}\label{prop:bounded}
  Suppose that $F$ has a compact support. More precisely, there exists a constant $C_0$, satisfying $$\Pr(\Vert \ttt_1^*\Vert \geq C_0) = 0,$$
  under the law of $F$.
  If we fix $\epsilon_{N,J}$ to be a constant $\epsilon$ independent of $N$ and $J$, satisfying
  \begin{equation}\label{eq:epschoice}
  0 < \epsilon \leq \frac{1}{2}\min\left\{1-f\left(C\sqrt{C_0^2+1}\right),f\left(-C\sqrt{C_0^2+1}\right), \frac{2}{5} \right\},
  \end{equation}
  then \eqref{eq:epsilon1} and \eqref{eq:epsilon2} are satisfied. This choice of $\epsilon_{N,J}$, together with the regularity conditions in Theorem~\ref{thm:consis},
  guarantees $L_{N, J}(A^*, \hat A)$ to converge to zero in probability. 
\end{proposition}

\begin{proposition}\label{prop:normal}
Consider exploratory IFA based on the M2PL model, where
 $F$ is a multivariate sub-Gaussian distribution\footnote{We say the distribution of a K-variate random vector $\ttt$ is sub-Gaussian, if
  there exist constant $b_1, b_2 >0$ such that for any $u \in \R^K, \|\mathbf u\|=1$ and $t>0$, $\Pr(|\mathbf u^\top\ttt| > t ) \leq b_1e^{-b_2t^2}$. In particular, the multivariate normal distribution is sub-Gaussian.} and $f$ is the logistic link. Suppose that there exists a constant $\beta \geq 1$ such that
 \begin{equation}\label{eq:NJrelation}
 J \leq N \leq J^{\beta}.
 \end{equation}
 Then 
 \eqref{eq:epsilon1} and \eqref{eq:epsilon2} hold, for any $\epsilon_{N,J}$ taking the form
 \begin{equation}\label{eq:eps_choice}
 \epsilon_{N,J} = \gamma_0 J^{-\gamma_1},
 \end{equation}
 where $\gamma_0$ and $\gamma_1$ are any constants satisfying $\gamma_0 > 0$ and $\gamma_1 \in (0, (4(K+3))^{-1})$.
The choice of $\epsilon_{N,J}$ following \eqref{eq:eps_choice}, together with the regularity conditions in Theorem~\ref{thm:consis}, guarantees $L_{N, J}(A^*, \hat A)$ to converge to zero in probability.

\end{proposition}

According to the result of Proposition~\ref{prop:bounded},
it suffices to choose $\epsilon_{N,J}$ as a sufficiently small positive constant, when $F$ has a bounded support.
Under the setting of Proposition~\ref{prop:normal}, to ensure consistency, one has to let $\epsilon_{N,J}$ decay to zero at an appropriate rate. 
Note that even in the second setting where the support of $F$ is unbounded, $\epsilon_{N,J}$ is almost like a constant, as it
decays to zero very slowly when $J$ grows. These results suggest that we may choose $\epsilon_{N,J}$ to be a sufficiently small constant in practice.

%
%
%
%

\paragraph{On the choice of $K$.} In the previous discussion, the number of factors $K$ is assumed to be known. In practice, however, this information is often unknown and an important task in exploratory IFA is to determine the number of factors based on data.
When conducting exploratory linear factor analysis, one typically gains the first idea by examining the scree plot from principal component analysis. Thanks to the connection between Algorithm~\ref{alg:SVD} and PCA as discussed in Remark~\ref{rmk:PCA},
a similar scree plot is available from the current method.

The scree plot is produced as follows. We first run Algorithm~\ref{alg:SVD}, but replace the unknown $K$ in step 1 of the algorithm by a reasonably large number $K^{\dagger}$.
Then, a scree plot can be obtained by plotting $\hat \sigma_k$ in a descending order, for $\hat \sigma_k$s produced by step 7 of Algorithm~\ref{alg:SVD}.
Figure~\ref{fig:fig3} shows such a scree plot, for which the data are generated from a five-factor model ($K=5$) with $J = 200$ and $N = 4000$, and the input number of factors is set to be $K^{\dagger} = 10$ in step 1 of the algorithm.
Unsurprisingly, an obvious gap is observed between $\hat \sigma_{5}$ and $\hat \sigma_{6}$. In fact, when data follow an IFA model, such a gap in the singular values is guaranteed to exist asymptotically, no matter what the input dimension is. In practice, the latent dimension $K$ can be chosen by identifying the singular value gap from the scree plot.

\begin{theorem}\label{thm:singular}
Under the same conditions as Theorem~\ref{thm:consis} and when the input dimension $K^{\dagger}$ in Algorithm~\ref{alg:SVD} is set fixed (i.e., independent of $N$ and $J$) but not necessarily equal to the true number of factors,
there exists a constant $\delta > 0$ such that for the true number of factors $K$,
$$\lim_{N, J\rightarrow \infty} \Pr\left(\frac{\hat \sigma_{K}}{\sqrt{NJ}} > \delta\right) = 1, \mbox{~and~~} \frac{\hat \sigma_{K+1}}{\sqrt{NJ}}  \overset{pr}{\to} 0,$$ as $N$ and $J$ grow to infinity simultaneously.

\end{theorem}

\begin{remark}
As shown in the proof, the input dimension $K^{\dagger}$ does not affect the asymptotics, as long as it does not grow with $N$ and $J$. However, for relatively small $N$ and $J$,
$X$ obtained in step 3 of the algorithm may not reserve enough information when the input dimension is smaller than $K+1$, which may lead to an
underestimation of the number of factors.
Thus, in practical applications,
we recommend to choose the input dimension to be slightly larger than the
maximum number of factors one suspects to exist in the data.
\end{remark}

 \begin{figure}
   \centering
   \includegraphics[scale = 0.7]{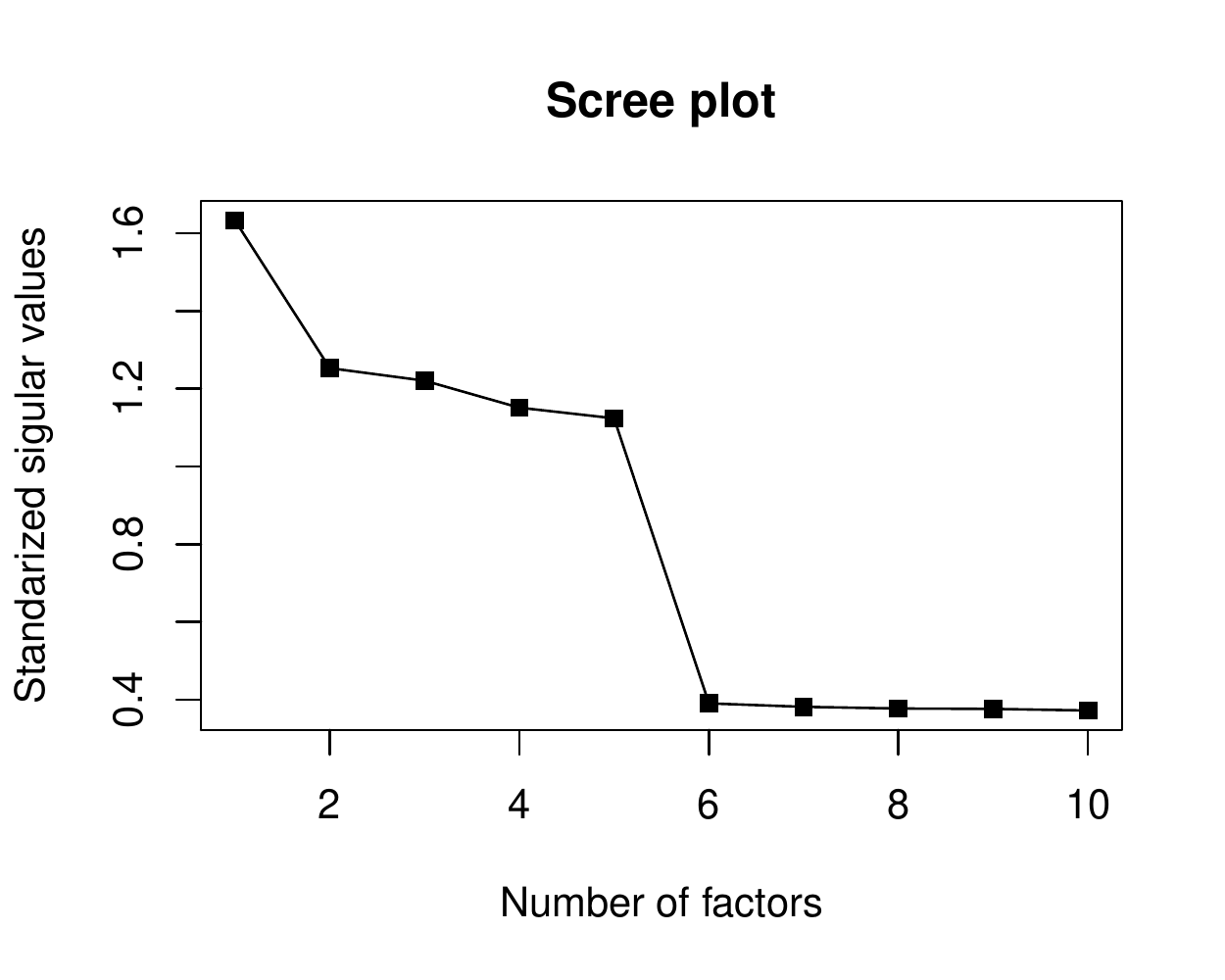}
   \caption{A scree plot for choosing the number of factors. The $y$-axis shows the standardized singular values $\hat \sigma_k/\sqrt{NJ}$, where $\hat \sigma_k$s are obtained from step 7 of Algorithm \ref{alg:SVD}. The data are simulated from an IFA model with $K=5$, $J = 200$, and $N = 4000$. The input dimension is set to be 10 in Algorithm~\ref{alg:SVD}. A singular value gap can be found between the 5th and 6th singular values.}\label{fig:fig3}
 \end{figure}

\paragraph{Statistical efficiency.}
We further point out that a price is paid for the computational advantage of the SVD-based estimator. To elaborate on this point, we compare it with the CJMLE \citep{chen2019joint,chen2017structured}.
The CJMLE treats both item parameters and latent factors as fixed parameters and maximizes a joint likelihood function with respect to all the fixed parameters.  
The SVD-based estimator is statistically less efficient than the CJMLE, in the sense that the SVD-based estimator converges to the true parameters in a much slower rate. To make this comparison, we consider
the same setting as in Proposition~\ref{prop:bounded}. The following proposition establishes the convergence rate for $\Vert X^* - \hat X\Vert_F^2/NJ,$
which determines the convergence of $\hat A$. Here, $X^* = (f(d_{j}^* + \aaaa_j^* (\ttt_i^*)^\top))_{N\times J}$ is the true item response probability matrix.

\begin{proposition}\label{prop:upper}
Suppose that the same assumptions as in Proposition~\ref{prop:bounded} hold and choose $\epsilon_{N,J}$ as in Proposition~\ref{prop:bounded}.
Then we have
\begin{equation}\label{eq:rate}
\frac{1}{NJ}\| X^* - \hat X \|_F^2  = O_p(J^{-\frac{1}{K+2}}).
\end{equation}

\end{proposition}
On the other hand, as shown in \cite{chen2017structured}, the CJMLE achieves the optimal rate (in minimax sense)  for estimating $X^*$, that is,
$\| X^* - \hat X_{JML} \|_F^2/(NJ)  = O_p(J^{-1}),$
where $\hat X_{JML}$ denotes the CJMLE.
This result suggests that the SVD-based estimator converges in a much slower rate than the CJMLE.



\section{Extensions}

\paragraph{Dealing with missing data.}
With slight modification, Algorithm~\ref{alg:SVD}  can handle item response data with missing values. We use matrix $W = (w_{ij})_{N\times J}$ to indicate the data nonmissingness, where $w_{ij} = 1$ indicates the response $Y_{ij}$ is not missing and $w_{ij} =0$ otherwise.
The modified algorithm is described as follows.
\begin{algorithm}[SVD-based estimator for exploratory IFA with missing data]\label{alg:SVD2}
\
\begin{enumerate}
\item Input nonmissing indicator $W = (w_{ij})_{N\times J}$, nonmissing responses $\{y_{ij}: w_{ij} =1, i = 1, ..., N, j = 1, ..., J\}$, the number of factors $K$, inverse link function $f$, and truncation parameter $\epsilon_{N,J} > 0$.
\item Compute $\hat p = (\sum_{i=1}^N\sum_{j=1}^J w_{ij})/(NJ)$ as the proportion of observed responses.
\item For each $i$ and $j$, let $z_{ij} = y_{ij}$, if $w_{ij} = 1$, and $z_{ij} = 0$ if $w_{ij} = 0$.
\item Apply the singular value decomposition to $Z$ to obtain $Z = \sum_{j = 1}^J \sigma_j \uu_j\vv_j^\top$, where $\sigma_1 \geq ...\geq \sigma_J \geq 0$ are the singular values and $\uu_j$s and $\vv_j$s are left and right singular vectors, respectively.

\item
 Let $$X = (x_{ij})_{N \times J} = \frac{1}{\hat p}\sum_{k = 1}^{\tilde K} \sigma_k \uu_k\vv_k^\top,$$
where $\tilde K = \max\big\{K+1, \argmax_k\{\sigma_k \geq 1.01 \sqrt{N(\hat p+3\hat p(1-\hat p))}\}\big\}$.
\item Let $\hat X = (\hat x_{ij})_{N\times J}$ be defined as
$$
\hat x_{ij} =
\begin{cases}
\epsilon_{N,J}, \quad \text{if } x_{ij} < \epsilon_{N,J},\\
x_{ij}, \quad \text{if } \epsilon_{N,J} \leq x_{ij} \leq 1-\epsilon_{N,J},\\
1-\epsilon_{N,J}, \quad \text{if } x_{ij} > 1 - \epsilon_{N,J}.
\end{cases}
$$
\item Let $\tilde M = (\tilde m_{ij})_{N\times J},$ where $\tilde m_{ij} = f^{-1}(\hat x_{ij}).$
\item Let $\hat\dd = (\hat d_1,...,\hat d_J)$, where $\hat d_j = (\sum_{i=1}^N\tilde m_{ij})/N$.
\item Apply singular value decomposition to $\hat M = (\tilde m_{ij} - \hat d_j)_{N\times J}$ to have $\hat M = \sum_{j = 1}^J \hat\sigma_j \hat\uu_j\hat\vv_j^\top$, where $\hat\sigma_1 \geq ...\geq \hat\sigma_J \geq 0$ are the singular values and $\hat\uu_j$s and $\hat\vv_j$ are the left and right singular vectors, respectively.
\item Output $\hat A = \frac{1}{\sqrt{N}}(\hat \sigma_1 \hat\vv_1,...,\hat\sigma_K\hat\vv_K), \hat \Theta = \sqrt{N}(\hat\uu_1,...,\hat\uu_K).$

\end{enumerate}
\end{algorithm}

\begin{remark}
It is easy to see that $\hat p = 1$ when there is no missing data. In that case, Algorithm~\ref{alg:SVD2} becomes exactly the same as Algorithm~\ref{alg:SVD}.
Steps 2-5 essentially follow the same procedure of \cite{chatterjee2015matrix} for matrix completion and the rest of the steps are the same as those in Algorithm~\ref{alg:SVD}.
Specifically, missing data are first imputed by zero in step 3 of the algorithm.  The bias brought by the simple imputation procedure is corrected in Step 5, by multiplying the factor $1/\hat p$. Similar to Algorithm~\ref{alg:SVD}, the choice of $\tilde K$ in step 5 is determined by the  procedure of \cite{chatterjee2015matrix} with a small modification which guarantees $\tilde K \geq K+1$.
\end{remark}

In fact, when the entries of the item response matrix are missing completely at random,
using a similar proof, one can show that $\hat A$ given by Algorithm~\ref{alg:SVD2} is still consistent,
under some mild condition on the missing data mechanism 
and the same conditions as in Theorem~\ref{thm:consis}. Specifically, the following condition is needed, in addition to conditions A1-A5.
\begin{itemize}
\item[A6.] The $w_{ij}$s are independent and identically distributed from a Bernoulli distribution with $\Pr(w_{ij} = 1) = p,$ where $0<p\leq1$ is a constant which does not depend on $N$ and $J$.
\end{itemize}

Under conditions A1-A6, the following proposition holds that guarantees the consistency of the proposed SVD estimator.

\begin{proposition}\label{prop:missing}
Under the same conditions as Theorem \ref{thm:consis} plus condition A6, the estimate $\hat A$ given by Algorithm~\ref{alg:SVD2} satisfies
 $L_{N,J}(A^*,\hat A) \overset{pr}{\to} 0$, as $N, J \rightarrow \infty.$
\end{proposition}


\paragraph{Dealing with ordinal data.} In exploratory IFA, ordinal data are also commonly encountered,
due to the wide use of Likert-scale items. With slight modification, the SVD method can also be used to analyze ordinal data. This is achieved by applying Algorithm~\ref{alg:SVD} to multiple dichotomized versions of data.

More precisely, consider data $Y = (Y_{ij})_{N\times J}$, where $Y_{ij} \in \{0, 1, ..., T\}$. We consider a general family of graded response type models, 
\begin{equation}\label{eq:ordinal}
\Pr(Y_{ij}\geq t \vert \ttt_i) = f(d_{jt} + \aaaa_j^\top \ttt_i),
\end{equation}
where $d_{jt}$ is an item- and category-specific intercept parameter,
and the rest of the notations are the same as that of model~\eqref{eq:IRF}. Note that the linear combination of the factors $\aaaa_j^\top \ttt_i$
does not depend on the response category and appears in all the submodels $\Pr(Y_{ij}\geq t \vert \ttt_i)$ for $t = 1, ..., T$.
When $f(x) = \exp(x)/(1+\exp(x))$ takes the logistic form,  model~\eqref{eq:ordinal}
becomes the multidimensional graded response model \citep{muraki1995full}.

Model~\eqref{eq:ordinal} is closely related to the general model \eqref{eq:IRF} for binary data. In fact, if we dichotomize data at response category $t$, i.e.,
$Y_{ij}^{(t)} = 1_{\{Y_{ij} \geq t\}}$,
then binary data $Y_{ij}^{(t)}$ follows model \eqref{eq:IRF} with the same loading parameters. Therefore, the loading matrix $A$ can be estimated by applying Algorithm~\ref{alg:SVD} to dichotomized data $Y^{(t)} = (1_{\{y_{ij}\geq t\}})_{N\times J}$, for some $t = 1, ..., T$. The estimation accuracy may be further improved by aggregating the results from multiple dichotomized versions of data. This aggregation method is summarized by Algorithm~\ref{alg:SVD3} below.

\begin{algorithm}[SVD-based estimator for exploratory IFA with ordinal data]\label{alg:SVD3}
\
\begin{enumerate}
\item Input response $Y = (y_{ij})_{N\times J}$, the number of categories $T$, the number of factors $K$, inverse link function $f$, and truncation parameter $\epsilon_{N,J} > 0$.
\item For $t = 1, ..., T$, apply Algorithm~\ref{alg:SVD} to dichotomized data $Y^{(t)} = (1_{\{y_{ij}\geq t\}})_{N\times J}$ and obtain
$\hat M^{(t)}$ from step 7 of Algorithm~\ref{alg:SVD}.
\item Let $\hat M = (\sum_{t=1}^T \hat M^{(t)})/{T}$. Apply singular value decomposition to $\hat M $ and obtain $\hat M = \sum_{j = 1}^J \hat\sigma_j \hat\uu_j\hat\vv_j^\top$, where $\hat\sigma_1 \geq ...\geq \hat\sigma_J \geq 0$ are the singular values and $\hat\uu_j$s and $\hat\vv_j$ are left and right singular vectors, respectively.
\item Output $\hat A = \frac{1}{\sqrt{N}}(\hat \sigma_1 \hat\vv_1,...,\hat\sigma_K\hat\vv_K), \hat \Theta = \sqrt{N}(\hat\uu_1,...,\hat\uu_K).$
\end{enumerate}
\end{algorithm}

\section{Simulation}\label{sec:sim}

\paragraph{Simulation setting.} We consider $K = 4$ and $8$, $J = 200, 400, 600, 800$, $1000$, and $1200$, and $N = 20J$. For each combination of $N$, $J$, and $K$, two different latent distributions $F$ are considered, one is a $K$-variate standard normal distribution, and the other is a $K$-variate normal distribution $N(\mathbf 0, (\sigma_{ij})_{K\times K})$, where
$\sigma_{ij} = 1$ if $i=j$ and $\sigma_{ij} = 0.3$ if $i\neq j$.
The inverse link $f$ is chosen to be logistic,
i.e. $f(x) = \exp(x)/(1+\exp(x))$.
This leads to 24 different simulation settings, for all possible combinations of
$N$, $J$, $K$, and $F$.

For each simulation setting, 100 independent replications are generated, with the item parameters keeping fixed across replications. When $J = 200$ and given $K$, the item parameters are generated as follows.
\begin{enumerate}
  \item $d_1^*$, ..., $d_{200}^*$ are i.i.d. from a uniform distribution over interval $[-1,1]$.
  \item $\aaaa_1^*$, ..., $\aaaa_{200}^*$ are i.i.d., with $\aaaa_j^* = (a_{j1}^\dagger q_{j1}, ..., a_{jK}^\dagger  q_{jK})^\top$. Here $a_{jk}^\dagger$s are i.i.d. from a uniform distribution over interval $[1,2]$, and $\mathbf q_{j} = (q_{j1}, ..., q_{jK})^\top$ are i.i.d. from a uniform distribution over $\mathcal Q_K$. Specifically,
      $$\mathcal Q_4 = \left\{(q_1, ..., q_4)^\top: q_k \in \{0,1\}, \sum_{k=1}^4 q_k \geq 1, \mbox{~and~} \sum_{k=1}^4 q_k \leq 3\right\},$$
      and
      $$\mathcal Q_8 = \left\{(q_1, ..., q_8)^\top: q_k \in \{0,1\}, \sum_{k=1}^8 q_k \geq 1, \mbox{~and~} \sum_{k=1}^8 q_k \leq 3 \right\}.$$
      The $\mathbf q_{j}$s lead to sparse loading vectors.
\end{enumerate}
When $J > 200$, we set the item parameters by repeating multiple times the parameters under $J = 200$ and the same $K$. For example, when $J = 400$, we set parameters for items 1-200 and those for items 201-400 to be the same as the parameters generated under the setting $J = 200$.

\paragraph{Results.} Each simulated dataset is analyzed using the SVD-based estimator, with the truncation parameter $\epsilon_{N,J}$ set to be $10^{-4}$. The performance of the SVD-based estimator is compared with that of the CJMLE.\footnote{The CJMLE is implemented using R package \emph{mirtjml} \citep{mirtjml}. All the computation is conducted on a single Intel\circledR Gold 6130 core.} 

\begin{figure}[h]
\centering
\begin{subfigure}[b]{0.3\textwidth}
\centering
\includegraphics[scale = 0.4]{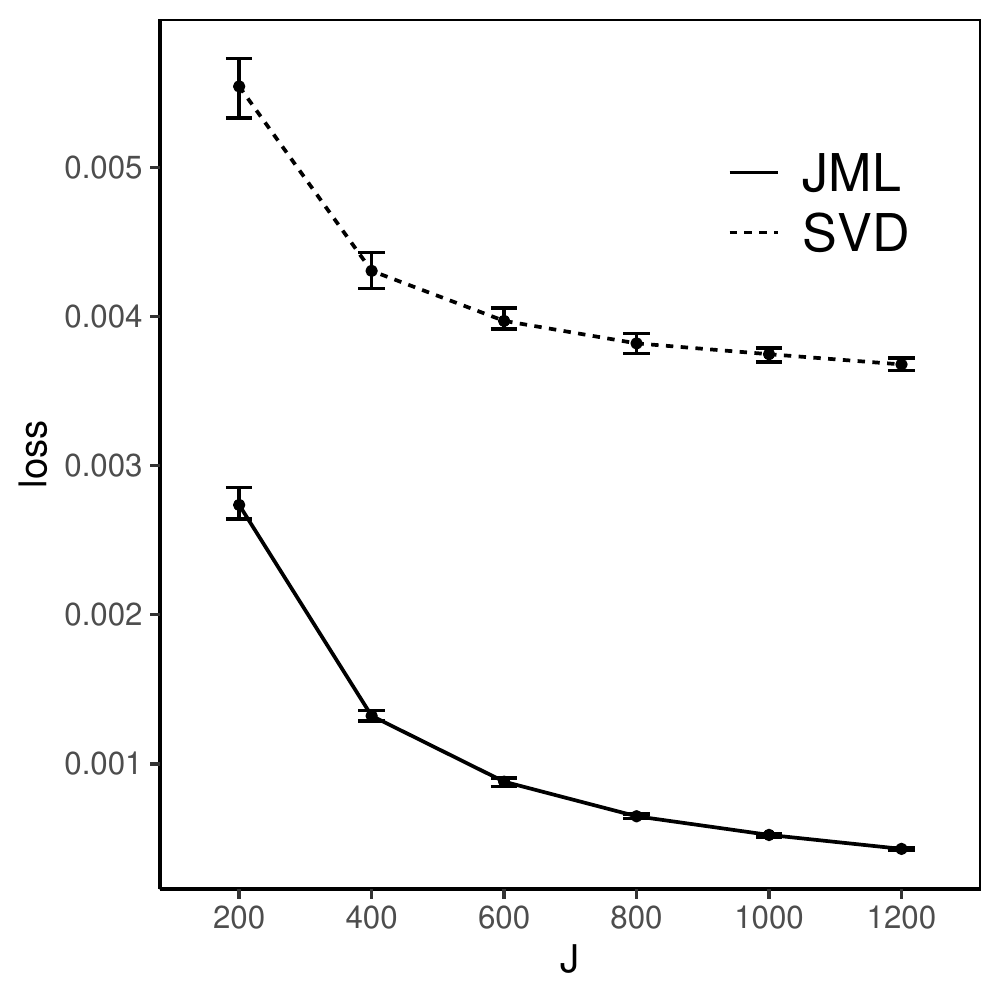}
\caption{}
\label{fig:fig1}
\end{subfigure}
\begin{subfigure}[b]{0.3\textwidth}
\includegraphics[scale = 0.4]{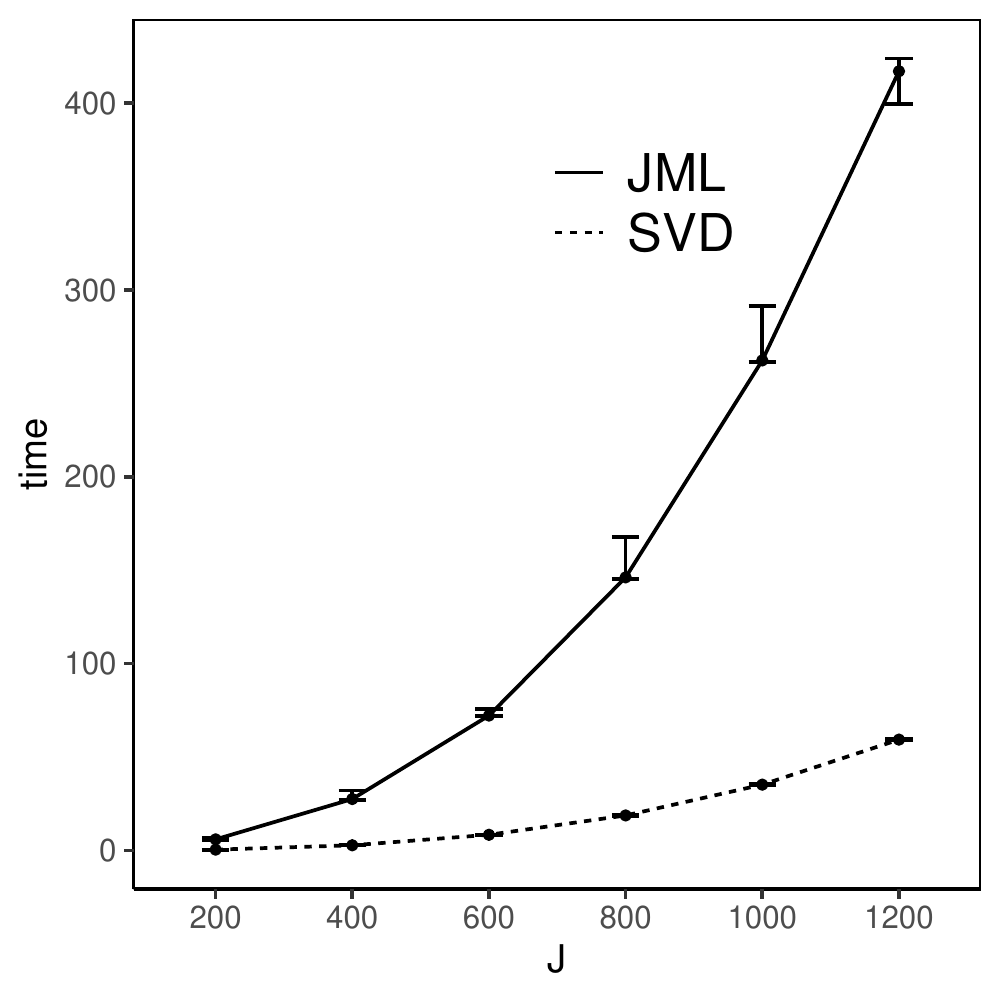}
\caption{}
\label{fig:fig2}
\end{subfigure}
\caption{\small Simulation results when $K=4$ and the true factors are independent. Panel (a) shows the number of items $J$ in $x$-axis versus the loss~\eqref{eq:ave} in $y$-axis and Panel (b) shows the number of items $J$ in $x$-axis versus the computation time (in seconds) in $y$-axis.
For each metric and each method, we show the median, 25\% quantile and 75\% quantile based on the 100 independent replications.}\label{fig:plot_K4_Gaussian_alpha0}
\end{figure}

\begin{figure}[h]
\centering
\begin{subfigure}[b]{0.3\textwidth}
\centering
\includegraphics[scale = 0.4]{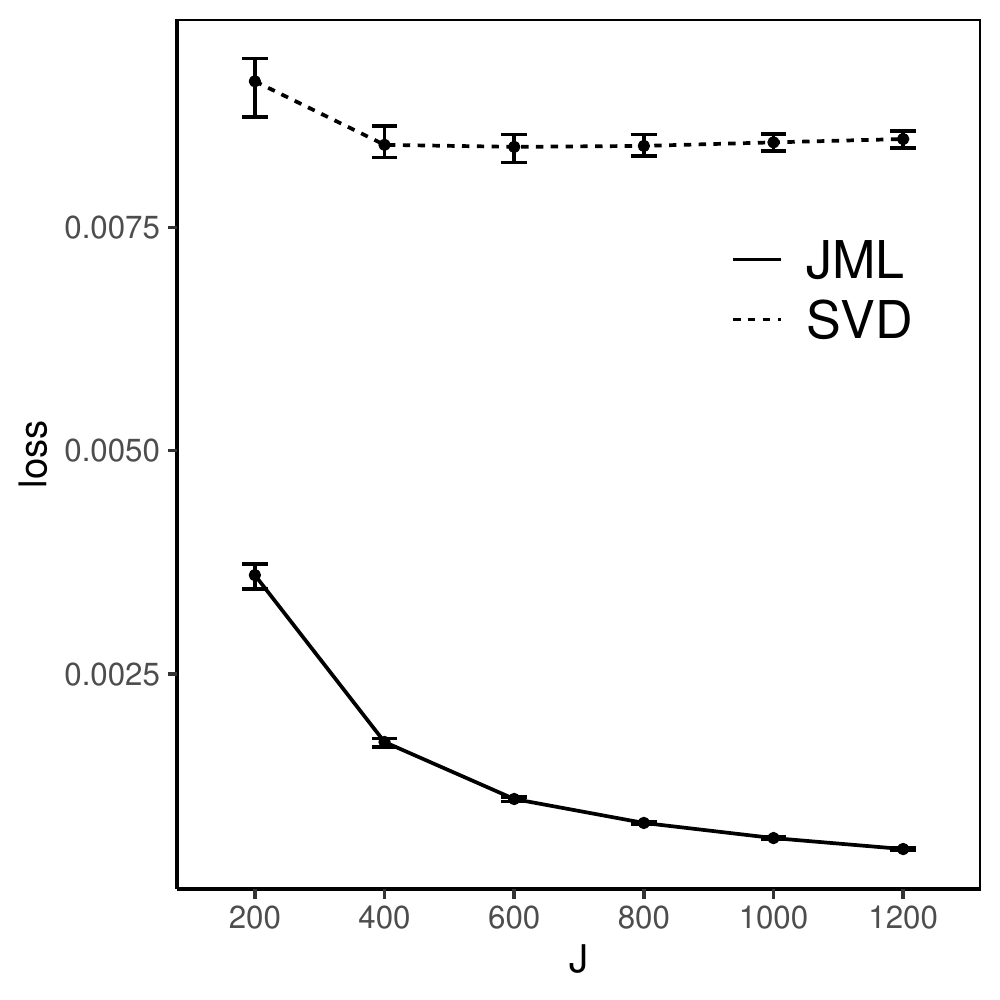}
\caption{}
\label{fig:fig1}
\end{subfigure}
\begin{subfigure}[b]{0.3\textwidth}
\includegraphics[scale = 0.4]{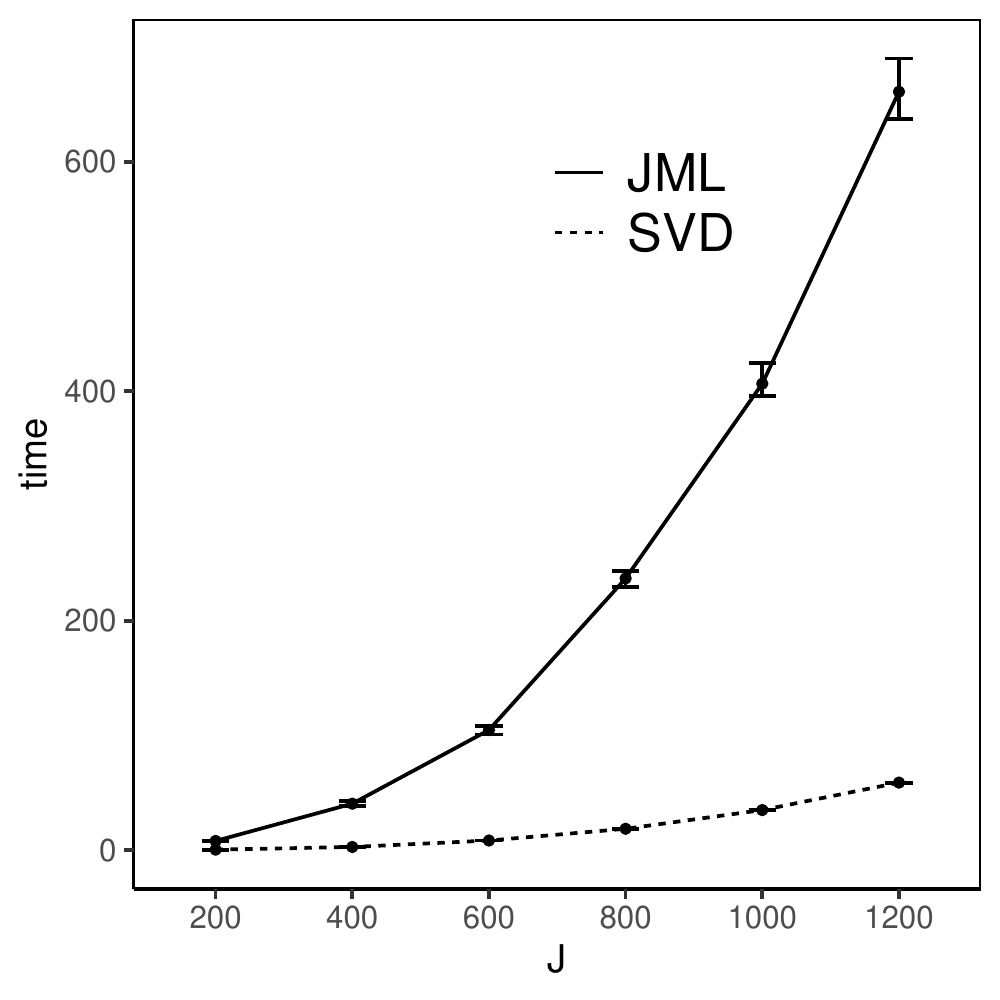}
\caption{}
\label{fig:fig2}
\end{subfigure}
\caption{\small Simulation results when $K=4$ and the true factors are correlated. The two panels show the same metrics as in Figure~\ref{fig:plot_K4_Gaussian_alpha0}.}\label{fig:plot_K4_Gaussian_alpha03}.
\end{figure}


\begin{figure}[h]
\centering
\begin{subfigure}[b]{0.3\textwidth}
\centering
\includegraphics[scale = 0.4]{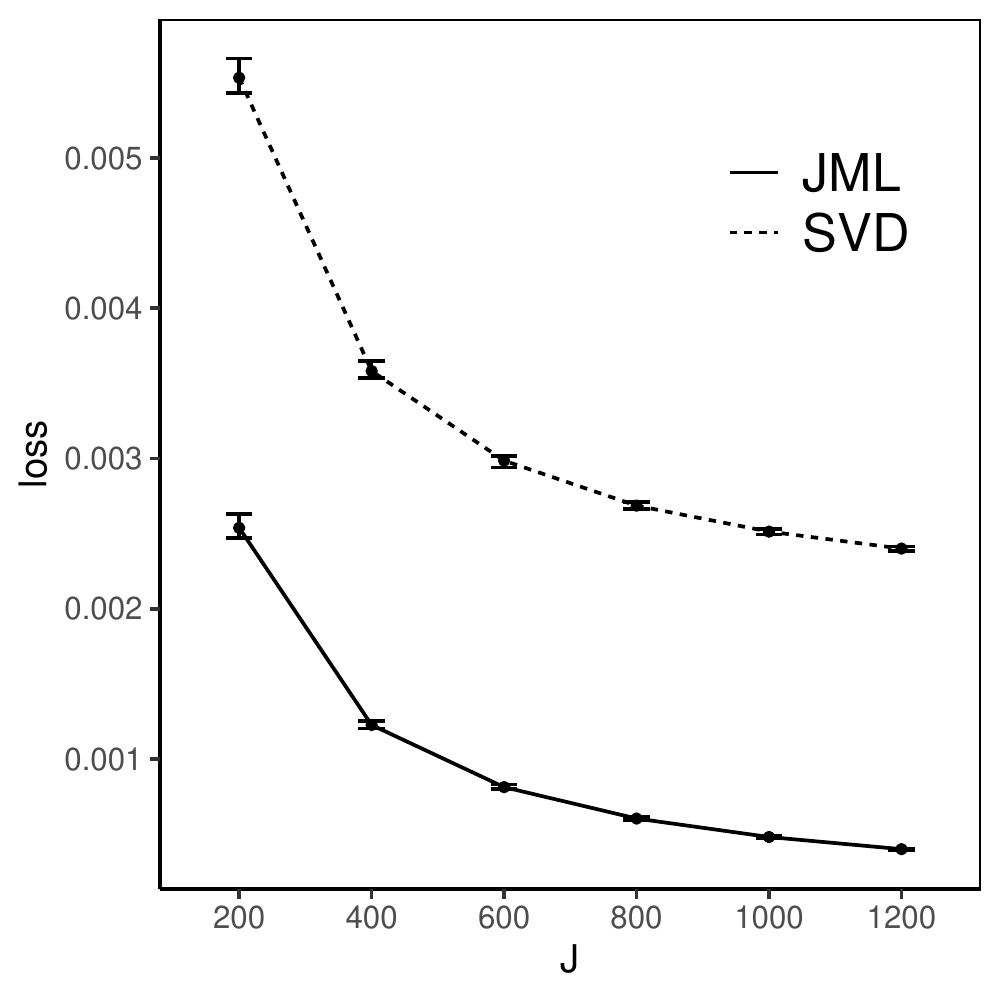}
\caption{}
\label{fig:fig1}
\end{subfigure}
\begin{subfigure}[b]{0.3\textwidth}
\includegraphics[scale = 0.4]{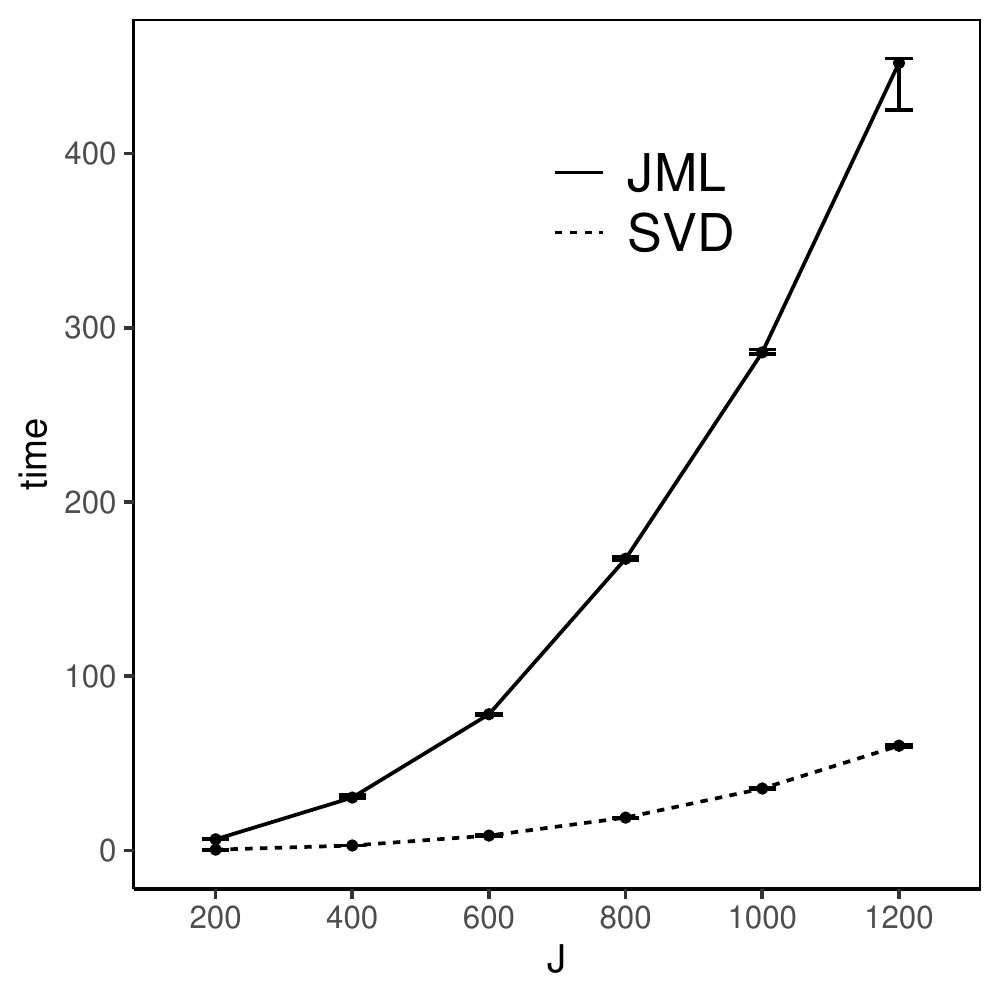}
\caption{}
\label{fig:fig2}
\end{subfigure}
\caption{\small Simulation results when $K=8$ and the true factors are independent. The two panels show the same metrics as in Figure~\ref{fig:plot_K4_Gaussian_alpha0}. }\label{fig:plot_K8_Gaussian_alpha0}
\end{figure}

\begin{figure}[h]
\centering
\begin{subfigure}[b]{0.3\textwidth}
\centering
\includegraphics[scale = 0.4]{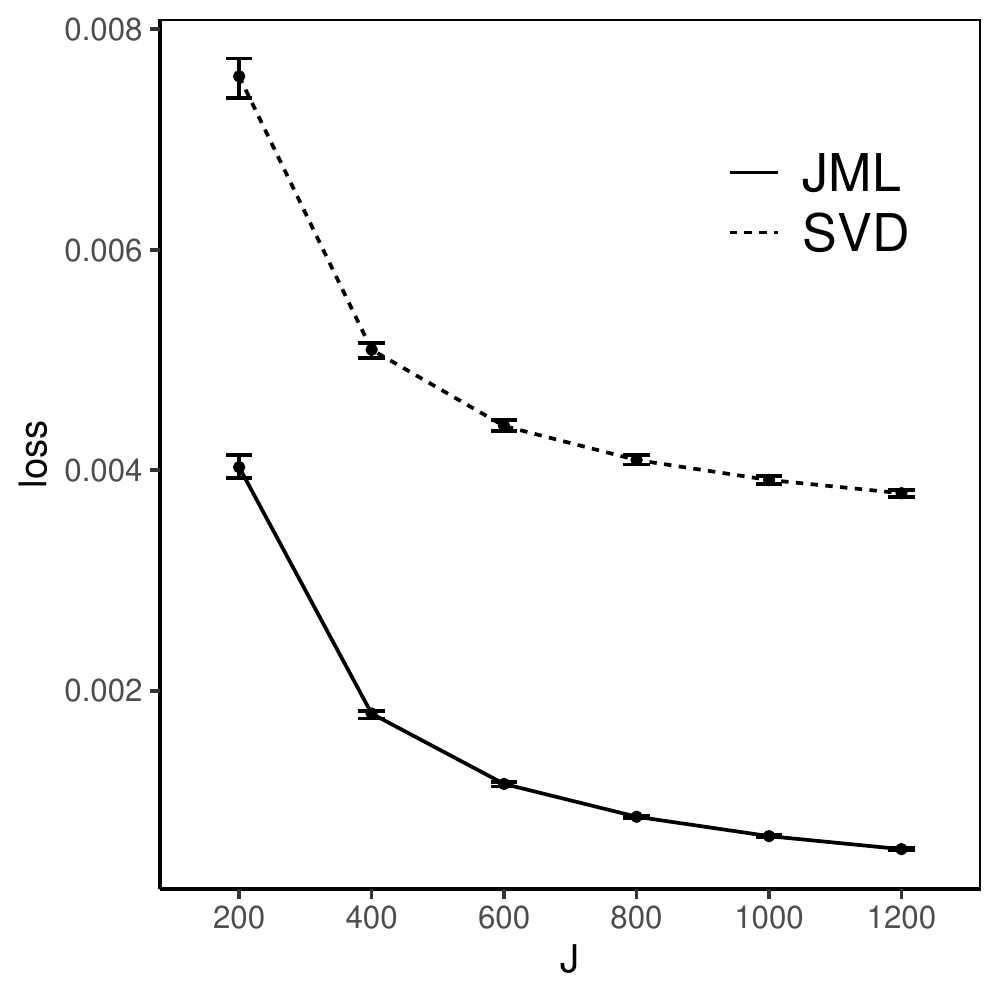}
\caption{}
\label{fig:fig1}
\end{subfigure}
\begin{subfigure}[b]{0.3\textwidth}
\includegraphics[scale = 0.4]{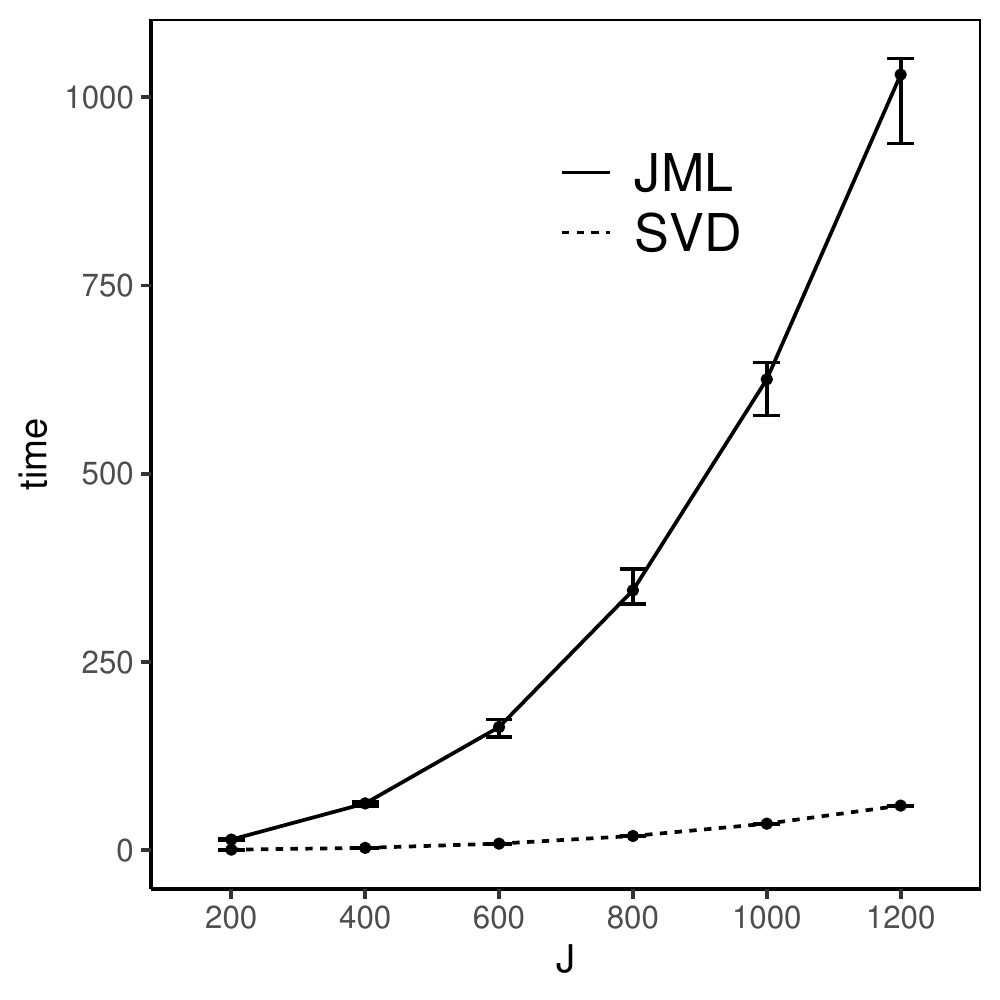}
\caption{}
\label{fig:fig2}
\end{subfigure}
\caption{\small Simulation results when $K=8$ and the true factors are correlated. The two panels show the same metrics as in Figure~\ref{fig:plot_K4_Gaussian_alpha0}.}\label{fig:plot_K8_Gaussian_alpha03}
\end{figure}

The loss for the SVD-based estimator decreases when $N$ and $J$ simultaneously grow, under all settings. Reasonable accuracy can be achieved when $N$ and $J$ are reasonably large, in which case the SVD-based estimator may be directly used for data analysis. For example, under the setting that $K = 4$ and $F$ is multivariate standard normal, the loss function is already around 0.006 when $J$ is 200. It suggests that the average entrywise error is around $0.08$.
In addition, the loss for the SVD-based estimator tends to be smaller when the factors are independent than that
when they are correlated, for the same $N$, $J$, and $K$. This is because, the signal in the data is weaker in the latter case, due to the redundant information in correlated factors.

Moreover, we compare the performance of the two estimators. The CJMLE is always more accurate than the SVD-based estimator. This is consistent with the asymptotic theory that the CJMLE is statistically more efficient. However, if we compare the computation time of the two approaches, the SVD-based estimator is substantially faster. Under the most time consuming setting where
$J = 1200, K = 8$ and the factors are correlated, the SVD approach only takes about 60 seconds, while the CJMLE takes about 17 minutes.
Note that as shown in \cite{chen2019joint},
CJMLE is already substantially faster than the marginal maximum likelihood estimator.
Given its reasonable accuracy and computational advantage, the SVD-based estimator may be a good alternative to the CJMLE and the MMLE in large-scale exploratory IFA problems.

\section{Concluding Remarks}\label{sec:con}

As shown in this note, the proposed SVD-based algorithm is statistically consistent and has good finite sample performance in large-scale
exploratory IFA problems. Although not statistically most efficient, the algorithm has its unique strengths over other exploratory IFA methods. In particular, it is computationally much faster. In addition, it guarantees a unique solution, while most of the other estimators can suffer from convergence issues for involving nonconvex optimization, including the CJMLE and MMLE.

Given its computational advantages and good finite sample performance, the SVD-based estimator can be used, not only as a starting point for other estimators to improve their numerical convergence, but also as an alternative estimator for data analysis. Specifically, in large-scale
exploratory IFA applications, we suggest to start data exploration with the SVD-based estimator. Using this estimator, we can
quickly gain some understanding about the number of factors underlying the data, and the loading structures of IFA models assuming different numbers of factors. Such initial knowledge helps us to focus on a smaller set of latent dimension $K$.
For these latent dimensions, we tend to further investigate their loading structures by the CJMLE, using the corresponding SVD solutions as starting points. When sample and item sizes are relatively smaller, the traditional methods may be more suitable, such as the MMLE and the composite-likelihood-based estimator. 




One limitation of the SVD-based estimator is that it is not easy to make statistical inference on the estimated loading matrix, such as constructing a confidence interval for an estimated loading parameter. 
This type of inference problem is not an issue for estimators based on the marginal likelihood, for which the asymptotic regime let
$N$ diverge and keep $J$ fixed.
However, it is a general challenge for both the SVD-based estimator and the CJMLE, whose consistency relies on a double asymptotic regime and the notion of consistency is weaker than  that in the traditional sense. In recent years, this type of inference problems has received much attention in statistics \citep{chen2019inference,xia2019statistical}. However, to the best of our knowledge,
no results have been obtained under an IFA model. We leave this problem for future investigation.


\vspace{3cm}

\noindent
{\LARGE Appendix}

\appendix

\numberwithin{equation}{section}

\section{Notations}\label{Sec:notation}

Let $\Theta^* = (\ttt_1^*,...,\ttt_N^*)^\top =  (\theta_{ik})_{N\times K}, A^* = (\aaaa_1^*,...,\aaaa_J^*)^\top =  (a_{jk})_{J\times K},$ and $\dd^* = (d_1^*,...,d_J^*)$ denote the true person parameters, factor loadings and intercept parameters, respectively.
We also denote
$\ttt_i^+ = (1,(\ttt_i^*)^\top)^\top, \aaaa_j^+ = (d_j^*, (\aaaa_j^*)^\top)^\top, \quad \text{for } i=1,...,N,\quad j = 1,...,J.$
We use  $\E_N, \0_N$ to denote $N$ dimensional vectors with all entries being $1$ and $0$ respectively, and
$B^{(K)}_{\aaaa}(C)$ to denote the ball in $\R^K$ centered at $\aaaa \in \R^K$ with radius $C.$
For a matrix $Z = (z_{ij})_{m \times n}$ and a function $f: \mathbb{R}\to\mathbb{R}$, let $f(Z) := (f(z_{ij}))_{m\times n}$.
Let $\sigma_k(Z)$ denote the $k$-th largest singular value of $Z,$ and $\|Z\|, \| Z \|_*$ denote the spectrum norm and nuclear norm of $Z$, which is the largest singular value and the sum of all singular values, respectively.
If $Z$ is a square matrix, let $\lambda_k(Z)$ denote the $k$-th largest eigenvalue of $Z.$

We denote $$X^* := (x_{ij}^*)_{N\times J} = f(\Theta^*(A^*)^\top + \E_N (\dd^*)^\top)$$ as the true probability matrix
and define $\tilde X = (\tilde x_{ij})_{N\times J}$ by
$$
\tilde x_{ij} =
\begin{cases}
0, \quad \text{if } x_{ij} < 0,\\
x_{ij}, \quad \text{if } 0 \leq x_{ij} \leq 1,\\
1, \quad \text{if } x_{ij} > 1,
\end{cases}
$$ where $x_{ij}$ is defined in step 5 of Algorithm \ref{alg:SVD2}.

Throughout the proof, we use $c$ to denote constant, whose value may change from line to line or even within a line.
We will drop the subscripts in $\epsilon_{N,J}$ and write $\epsilon$ for notional simplicity.

\section{Proof of Theorems}
\begin{proof}[Proof of Theorem \ref{thm:consis}]
Since Theorem \ref{thm:consis} is a special case of Proposition \ref{prop:missing} when $p = 1$ and $W = \E_N\E_J^\top,$ we refer the readers to the proof of Proposition \ref{prop:missing}.
\end{proof}

\begin{proof}[Proof of Theorem \ref{thm:singular}]
Let $\sigma^*_k$ denote the $k$th largest singular value of $\Theta^*(A^*)^\top.$ Then we have
\begin{equation} \label{eq:singular temp1}
|\hat \sigma_k - \sigma^*_k| \leq \|\hat M - \Theta^*(A^*)^\top \|_2 \leq \| \hat M - \Theta^*(A^*)^\top \|_F.
\end{equation}
By \eqref{eq:ineq1} in the proof of Lemma \ref{lem:P consis}, we can get
\begin{equation}\label{eq:singular temp2}
\frac{1}{\sqrt{NJ}} \| \hat M - \Theta^*(A^*)^\top \|_F \overset{pr}{\to} 0.
\end{equation}
Notice that \eqref{eq:singular temp2} holds as long as the input dimension in the algorithm is fixed.
Combine \eqref{eq:singular temp1} and \eqref{eq:singular temp2} to have
$$\frac{|\hat \sigma_k - \sigma_k^*|}{\sqrt{NJ}} \overset{pr}{\to} 0.$$
Notice that $\sigma_{K+1}^* = 0$ and we get $$\frac{\hat \sigma_{K+1}}{\sqrt{NJ}}  \overset{pr}{\to} 0.$$ For $k = K,$ we get $$\Pr\left( \frac{|\hat \sigma_K - \sigma_K^*| }{\sqrt{NJ}} \leq \tilde\epsilon \right) \to 1$$ for any $\tilde\epsilon>0$ and thus
\begin{equation}\label{eq:sigma ineq}
\Pr\left( \frac{\hat \sigma_K }{\sqrt{NJ}} \geq \frac{1}{\sqrt{NJ}} \sigma^*_K - \tilde\epsilon \right) \to 1.
\end{equation}
For $\sigma_K^*,$ we have
\begin{equation}\label{eq:sigma bound}
\begin{aligned}
\frac{1}{\sqrt{NJ}} \sigma_K^* &= \frac{1}{\sqrt{NJ}}\sigma_K(\Theta^*(A^*)^\top)\\
&\geq \frac{1}{\sqrt{N}}\sigma_K(\Theta^*) \frac{1}{\sqrt{J}}\sigma_K(A^*)\\
&\geq C_1 \frac{1}{\sqrt{N}}\sigma_K(\Theta^*).
\end{aligned}
\end{equation}
The last inequality is due to condition A4.
Let $\hat \Sigma = \frac{1}{N}\sum_{i=1}^N \ttt_i^* (\ttt_i^*)^\top$ and it is not hard to verify that $$\frac{1}{\sqrt{N}}\sigma_K(\Theta^*) = \sqrt{\lambda_K(\hat \Sigma)}.$$ By law of large number, we know $$\| \hat \Sigma - \Sigma^* \|_2 \overset{pr}{\to} 0$$ which leads to $$\lambda_K(\hat \Sigma) \overset{pr}{\to} \lambda_K(\Sigma^*) > 0$$ and thus
\begin{equation}\label{eq:sigma theta bound}
\frac{1}{\sqrt{N}}\sigma_K(\Theta^*) = \sqrt{\lambda_K(\hat \Sigma)} \overset{pr}{\to} \sqrt{\lambda_K(\Sigma^*)} > 0.
\end{equation}
Combining \eqref{eq:sigma ineq}, \eqref{eq:sigma bound}, \eqref{eq:sigma theta bound} and choosing $\tilde\epsilon = \frac{1}{2}\sqrt{\lambda_K(\Sigma^*)},$ we have $$\Pr\left( \frac{\hat \sigma_K }{\sqrt{NJ}} > \frac{1}{4} \sqrt{\lambda_K(\Sigma^*)} \right) \to 1.$$ We complete the proof by choosing $\delta = \frac{1}{4} \sqrt{\lambda_K(\Sigma^*)}.$
\end{proof}

\section{Proof of Propositions}
\begin{proof}[Proof of Proposition \ref{prop:bounded}]
According to the choice of $\epsilon$, we have
$h(2\epsilon) \geq C\sqrt{C_0^2+1}.$
Then,	
\begin{align*}
\Pr\left( \|\ttt_1^*\| \geq h(2\epsilon)/C  \right) = 0,
\text{ and } \frac{(h(2\epsilon_{N,J}))^{\frac{K+1}{K+3}}}{(\epsilon_{N,J} g(\epsilon_{N,J}))^2} = O(1).
\end{align*}
We complete the proof by Theorem~\ref{thm:consis}.
\end{proof}

\begin{proof}[Proof of Proposition \ref{prop:normal}]

For the logistic link function, we have
\begin{align}\label{eq:h-g-logist}
h(y) = \log \frac{1-y}{y} \text{ and }
g(y) = y(1-y), \quad \text{for } y \in (0,0.5).
\end{align}
Since $\ttt_1^*$ is a sub-Gaussian random vector, then $\|\ttt_1^*\|_2^2$ is an sub-exponential random variable, which means there exist constant $c_1, c_2 > 0,$ such that for any $t > 0,$ we have $$\Pr(\|\ttt_1^*\|^2 \geq  t) \leq c_1\exp(-c_2t).$$
Then,
\begin{equation}\label{eq:to-verify}
\begin{aligned}
\Pr\left(\|\ttt_1^*\| \geq \frac{h(2\epsilon)}{C} \right) &= \Pr\left( \|\ttt_1^* \|^2 \geq \frac{(h(2\epsilon))^2}{C^2} \right) \\
&= \Pr\left( \|\ttt_1^* \|^2  \geq \frac{\log^2\left( \frac{1-2\epsilon}{2\epsilon} \right)}{C^2} \right)\\
&\leq c_1\exp\left( -c_2\frac{\log^2\left( \frac{1-2\epsilon}{2\epsilon} \right)}{C^2} \right)
\end{aligned}
\end{equation}


Recall we choose $\epsilon=\gamma_0 J^{-\gamma_1}$ in \eqref{eq:eps_choice}. Consequently,
\begin{equation}
\log^2\left(\frac{1-2\epsilon}{2\epsilon}\right)
= \gamma_1^2(\log J)^2 + O(\log J).
\end{equation}
Therefore,
\begin{equation}
	\begin{split}
		c_1\exp\left( -c_2\frac{\log^2\left( \frac{1-2\epsilon}{2\epsilon} \right)}{4C^2} \right)
		= &c_1\exp\left(
		- c_2\frac{\gamma_1^2(\log J)^2}{C^2} + O(\log J)
		\right)\\
		\leq & c_1\exp\left(
		- c_2\frac{\gamma_1^2(\log N)^2}{C^2\beta^2} + O(\log N)
		\right)\\
		= & o\left(\frac{1}{N}\right)
	\end{split}
\end{equation}
where the second inequality is due to the assumption that $J^{\beta}\geq N$. The above display together with \eqref{eq:to-verify} verifies \eqref{eq:epsilon1}. We proceed to verify \eqref{eq:epsilon2}.
According to \eqref{eq:h-g-logist}, we have
\begin{align*}
\frac{(h(2\epsilon))^{\frac{K+1}{K+3}}}{(\epsilon g(\epsilon))^2} &= \frac{\left( \log\left( \frac{1-2\epsilon}{2\epsilon} \right) \right)^{\frac{K+1}{K+3}}}{\epsilon^4(1-\epsilon)^2}.
\end{align*}
Plugging in $\epsilon=\gamma_0 J^{-\gamma_1}$, the above equation becomes
\begin{equation*}
	\frac{\left( \log\left( \frac{1-2\epsilon}{2\epsilon} \right) \right)^{\frac{K+1}{K+3}}}{\epsilon^4(1-\epsilon)^2}
	= \left(1+o(1)\right)\gamma_0^{-4}J^{4\gamma_1}\left(
	\gamma_1\log J + O(1)
	\right)^{\frac{K+1}{K+3}}
	= J^{4\gamma_1 + o(1)}.
\end{equation*}
Thus, for {$\gamma_1\in(0,\frac{1}{4(K+3)})$},
\begin{equation}
	\frac{\left( \log\left( \frac{1-2\epsilon}{2\epsilon} \right) \right)^{\frac{K+1}{K+3}}}{\epsilon^4(1-\epsilon)^2} = o(J^{\frac{1}{K+3}})
\end{equation}
This verifies \eqref{eq:epsilon2} and completes the proof by applying Theorem \ref{thm:consis}.

\end{proof}

\begin{proof}[Proof of Proposition \ref{prop:upper}]
The proof of Proposition \ref{prop:upper} is similar to proof of Lemma \ref{lem:P consis}. We will only state the main steps and omit the repeating details.
According to Lemma \ref{lem:chatterjee} in Appendix \ref{sec:lemma-proof}, we have
\begin{align}\label{eq:chatterjee 2}
\frac{1}{NJ}\EEE \left( \| \tilde X - X^* \|_F^2  \mid X^* \right) \leq c \min \left\{ \frac{\|X^*\|_*}{J\sqrt{N}}, \frac{\|X^*\|^2_*}{NJ}, 1 \right\}  + c e^{-cN},
\end{align}
Recall that we assume $\|\ttt^*_i\| \leq C_0$.
Following the similar arguments as in the proof of Lemma \ref{lem:P consis}, we have
\begin{align*}
\frac{1}{NJ}\EEE \left( \| \tilde X - X^* \|_F^2  \mid X^* \right) \leq c\frac{1}{\sqrt{J}}\left( \frac{C_0}{\delta} \right)^{\frac{K}{2}} + L\delta\left(\sqrt{C_0^2+1}+C\right) + c\exp(-cN).
\end{align*}
There is a difference from the proof of Lemma \ref{lem:P consis} that the rank of matrix $f(M_\delta)$ is upper bounded by
\begin{align*}
\mathrm{rank}(f(M_\delta)) \leq \min\{|\mathcal G_1|, |\mathcal G_2| \} \leq |\mathcal G_1| \leq c\left( \frac{C_0}{\delta} \right)^K.
\end{align*}
Choose $\delta = \left(\frac{cC_0^K}{JL^2\left(\sqrt{C_0^2+1}+C\right)^2}\right)^{\frac{1}{K+2}}$, then
\begin{align*}
\frac{1}{NJ}\EEE \left( \| \tilde X - X^* \|_F^2  \Big| X^* \right) \leq cJ^{-\frac{1}{K+2}} + c\exp\left(-cN\right).
\end{align*}
Let $g(N,J) := cJ^{-\frac{1}{K+2}} + c\exp\left(-cN\right)$. By taking expectation, we have
\begin{align*}
\frac{1}{NJ}\EEE \left( \| \tilde X - X^* \|_F^2 \right) \leq g(N,J).
\end{align*}
For any $\Delta_{N,J}>0$, by Chebyshev's inequality, we have
\begin{equation*}
	\Pr\left( \frac{1}{NJ} \| \tilde X - X^* \|_F^2 \geq  \frac{g(N,J)}{\Delta_{N,J}} \right)\leq \Delta_{N,J}.
\end{equation*}
Thus, for any sequence $\Delta_{N,J}$ satisfying $\Delta_{N,J}=o(1)$, we have
\begin{align*}
\lim_{N,J\to\infty}\Pr\left( \frac{1}{NJ} \| \tilde X - X^* \|_F^2 \leq \frac{g(N,J)}{\Delta_{N,J}} \right)=1.
\end{align*}
In what follows, we restrict our analysis to the event $\left\{\frac{1}{NJ} \| \tilde X - X^* \|_F^2 \leq \frac{g(N,J)}{\Delta_{N,J}}\right\}$.
By \eqref{eq:epschoice}, we have $x^*_{ij} = f((\ttt_i^*)^\top\aaaa_j^* + d_j^*) \in [2\epsilon, 1-2\epsilon],$ which leads to
\begin{align*}
\frac{1}{NJ} \sum_{i,j} 1_{\{ \tilde x_{ij} \notin [\epsilon,1-\epsilon] \}} \leq \frac{g(N,J)}{\Delta_{N,J}\epsilon^2}=\frac{1}{\Delta_{N,J}\epsilon^2}\left(cJ^{-\frac{1}{K+2}} + c\exp\left(-cN\right)\right).
\end{align*}
Following the similar procedure as in proof of Lemma \ref{lem:P consis}, we can further bound $\| \hat X - X^* \|_F^2$ by
\begin{align}
\frac{1}{NJ} \| \hat X - X^* \|_F^2 &\leq \frac{1}{\epsilon^2\Delta_{N,J}}\left( cJ^{-\frac{1}{K+2}} + c\exp(-cN) \right)\nonumber\\
&= \frac{1}{\Delta_{N,J}}\left( cJ^{-\frac{1}{K+2}} + c\exp(-cN) \right). \label{eq:rate1}
\end{align}
To summarize, we have
\begin{align*}
\Pr\left( \frac{1}{NJ} \| \hat X - X^* \|_F^2 \leq \frac{1}{\Delta_{N,J}}\left( cJ^{-\frac{1}{K+2}} + c\exp(-cN) \right) \right) \to 1, \quad \text{as } N,J \to \infty,
\end{align*}
for any $\Delta_{N,J} = o(1)$. This implies $\frac{1}{NJ} \| \hat X - X^* \|_F^2 = O_p\left(J^{-\frac{1}{K+2}} +\exp(-cN) \right)=O_p(J^{-\frac{1}{K+2}} )$, where the second equation is due to $N\geq J$.
\end{proof}

\begin{proof}[Proof of Proposition \ref{prop:missing}]

We have
\begin{equation*}
\begin{aligned}
L_{N,J}(A^*, \hat A) &= \frac{1}{JK} \min_{O\in \mathbb R^{K\times K}} \left\{ \Vert A^* - \hat A O \Vert_F^2 \right\} \\
&= \frac{1}{JK} \min_{O\in \mathbb R^{K\times K}} \left\{ \Vert (A^*\Sigma^{\frac{1}{2}})\Sigma^{-\frac{1}{2}} - \hat A O \Vert_F^2 \right\} \\
&\leq \frac{1}{JK} \min_{O\in \mathbb R^{K\times K}} \left\{ (\Vert A^*\Sigma^{\frac{1}{2}} - \hat A O\Sigma^{\frac{1}{2}} \Vert_F^2 \right\} \| \Sigma^{-\frac{1}{2}} \|_F^2\\
&= \frac{1}{JK} \min_{Q\in \mathbb R^{K\times K}} \left\{ (\Vert A^*\Sigma^{\frac{1}{2}} - \hat A Q \Vert_F^2 \right\} \| \Sigma^{-\frac{1}{2}} \|_F^2\\
&= \frac{1}{JK} \min_{Q\in \mathbb R^{K\times K}} \left\{ (\Vert \tilde A - \hat A Q \Vert_F^2 \right\} \| \Sigma^{-\frac{1}{2}} \|_F^2\\
&=L_{N,J}(\tilde A, \hat A) \|\Sigma^{-\frac{1}{2}}\|_F^2,
\end{aligned}
\end{equation*}
where $\tilde A = A^* \Sigma^{\frac{1}{2}}.$
Let $\tilde \Theta = (\tilde \ttt_1,...,\tilde \ttt_N)^\top = \Theta^* \Sigma^{-\frac{1}{2}}$.
Then  $\Theta^*(A^*)^\top = \tilde \Theta \tilde A^\top$, and
$\tilde \ttt_i$s are independent and identically distributed from a distribution $\tilde F$ which has mean $\0$ and covariance matrix $I_K$.
Therefore, 
it suffices to show $L_{N,J}(A^*, \hat A)  \overset{pr}{\to} 0$ when $\Sigma = I_K$.
We prove it through
the following two lemmas whose proofs are given in Appendix~\ref{sec:lemma-proof}.

\begin{lemma}\label{lem:P consis}
Assume conditions A1, A2, A3, A5 and A6 are satisfied and further assume that \eqref{eq:epsilon1} and \eqref{eq:epsilon2} are satisfied.
Then,
\begin{align*}
\frac{1}{NJ} \left\| \hat\Theta\hat A^\top - \Theta^*(A^*)^\top \right\|_F^2 \overset{pr}{\rightarrow} 0.,
\end{align*}
where $\hat\Theta$ and $\hat A$ are given in Algorithm \ref{alg:SVD2}.
\end{lemma}

\begin{lemma}\label{lem:A consis}
Suppose conditions A1, A2 and A4 are satisfied and further suppose that
\begin{align*}
\frac{1}{NJ} \left\| \hat\Theta\hat A^\top - \Theta^*(A^*)^\top \right\|_F^2 \overset{pr}{\rightarrow} 0.
\end{align*}
Then, $L_{N,J}(A^*,\hat A)  \overset{pr}{\rightarrow} 0.$
\end{lemma}
We complete the proof.
\end{proof}

\section{Proof of Lemmas}\label{sec:lemma-proof}
\begin{proof}[Proof of Lemma \ref{lem:P consis}]
We first give a lemma regarding the error bound for recovering the probability matrix $X^*.$

\begin{lemma}\label{lem:chatterjee}
Given $X^*$, we have
\begin{align}\label{eq:chatterjee}
\frac{1}{NJ}\EEE \left( \| \tilde X - X^* \|_F^2  \Big| X^* \right) \leq c \min \left\{ \frac{\|X^*\|_*}{J\sqrt{N}}, \frac{\|X^*\|^2_*}{NJ}, 1 \right\}  + c e^{-cN}.
\end{align}
\end{lemma}

Let
\begin{align*}
p_\epsilon := \Pr( \|\ttt_1^*\| > C_{\epsilon}),
\end{align*}
where {$C_{\epsilon}=h(2\epsilon)/C$} is a quantity depending on $\epsilon$.
Let
\begin{align*}
\mathcal A_{N,J} := \left\{  \| \ttt_i^* \| \leq C_{\epsilon}, \text{ for } i=1,...,N \right\}.
\end{align*}
{Then, according to the condition \eqref{eq:epsilon1}
\begin{align*}
\lim_{N,J\to\infty}\Pr(\mathcal A_{N,J}) = \lim_{N,J\to\infty}(1-p_\epsilon)^N = 1.
\end{align*}}
In what follows, we restrict the analysis to  the event $\mathcal A_{N,J}$.
Let $\mathcal G_1, \mathcal G_2$ be two $\delta$-nets for $B^{(K)}_0(C_{\epsilon})$ and $B^{(K+1)}_0(C),$ respectively. This means $\mathcal G_1 \subset B_0^{(K)}(C_{\epsilon}), \mathcal G_2 \subset B_0^{(K+1)}(C)$ and $$B^{(K)}_0(C_{\epsilon}) \subset \bigcup_{\xx \in \mathcal G_1}B_{\xx}^{(K)}(\delta), \quad  B^{(K+1)}_0(C) \subset \bigcup_{\xx \in \mathcal G_2}B_{\xx}^{(K+1)}(\delta).$$
For any $\ttt_i^*,$ let $p(\ttt_i^*)$ be a point in $\mathcal G_1$ such that
\begin{align*}
\| \ttt_i^* - p(\ttt_i^*) \| \leq \delta,
\end{align*}
which implies
\begin{align*}
\| \ttt_i^+ - (1,p(\ttt_i^*)^\top)^\top \| = \| \ttt_i^* - p(\ttt_i^*) \| \leq \delta.
\end{align*}
With a little abuse of notation, we use $p(\ttt_i^+)$ to denote $(1,p(\ttt_i^*)^\top)^\top.$
For any $\aaaa_j^+,$ let $p(\aaaa_j^+)$ be a point in $\mathcal G_2$ such that
\begin{align*}
\| \aaaa_j^+ - p(\aaaa_j^+) \| \leq \delta.
\end{align*}
It is not hard to see that we can find such $\mathcal G_1, \mathcal G_2$ such that
\begin{align*}
| \mathcal G_1 | \leq c \left( \frac{C_{\epsilon}}{\delta} \right)^{K}, \quad | \mathcal G_2 | \leq c \left( \frac{C}{\delta}  \right)^{K+1},
\end{align*}
This is due to definition of $\mathcal A_{N,J}$ and condition A1.
Let $M_{\delta} = (_{\delta}m_{ij} )_{N\times J}$, where $_{\delta} m_{ij} = f\left( p(\ttt_i^+)^\top p(\aaaa_j^+) \right),$ then we have
\begin{align*}
\mathrm{rank}(M_\delta) \leq \min\{|\mathcal G_1|, |\mathcal G_2| \} \leq |\mathcal G_2| \leq c \left( \frac{C}{\delta} \right)^{K+1}.
\end{align*}
Now we provide an upper bound for $\|X^*\|_*$ on the right-hand side of \eqref{eq:chatterjee}. We have
\begin{align}\label{eq:two}
\|X^*\|_* = \| f(M^*) \|_* \leq \overbrace{\| f(M^*) - f(M_{\delta}) \|_*}^{(I)} + \overbrace{\| f(M_{\delta}) \|_*}^{(II)}.
\end{align}
The second term on the right-hand side of the above display is bounded above by
\begin{align}\label{eq:II}
(II) \leq \sqrt{\mathrm{rank}(f(M_{\delta}))} \cdot \|f(M_\delta)\|_F \leq c \left( \frac{C}{\delta} \right)^{\frac{K+1}{2}} \sqrt{NJ}.
\end{align}
Now we consider the first term. We have
\begin{align*}
\left| (\ttt_i^+)^\top \aaaa_j^+ - (p(\ttt_i^+))^\top p(\aaaa_j^+) \right| &\leq \left| ( \ttt_i^+ )^\top (\aaaa_j^+ - p(\aaaa_j^+)) \right| + \left|  ( \ttt_i^+ - p(\ttt_i^+) )^\top p(\aaaa_j^+)  \right| \\
&\leq \sqrt{C_{\epsilon}^2+1}\cdot\delta + \delta C.
\end{align*}
So
\begin{align*}
| f(m^*_{ij}) - f(_\delta m_{ij}) | &= | f \left(        (  \ttt_i^+  )^\top  \aaaa_j^+ \right ) - f((p(\ttt_i^+))^\top p(\aaaa_j^+))  | \\
&\leq L\delta \left(\sqrt{C_{\epsilon}^2+1} + C \right).
\end{align*}
We have used the Lipschitz continuity in condition A3 here.
Then the first term in \eqref{eq:two} is bounded from above as
\begin{align}\label{eq:I}
(I) \leq \sqrt{J} \| f(M^*) - f(M_\delta) \|_F \leq L\delta  \left(\sqrt{C_{\epsilon}^2+1} + C\right) \sqrt{J}\sqrt{NJ}.
\end{align}
Here we used the fact that the rank of the matrix $f(M^*) - f(M_\delta)$ cannot exceed $J$ according to condition A5.
Combined \eqref{eq:chatterjee}, \eqref{eq:two}, \eqref{eq:II} and \eqref{eq:I}, then on the event $\mathcal A_{N,J},$
\begin{align*}
\frac{1}{NJ}\EEE \left( \| \tilde X - X^* \|_F^2  \Big| X^* \right) \leq c\frac{1}{\sqrt{J}}\left( \frac{C}{\delta} \right)^{\frac{K+1}{2}} + L\delta  \left(\sqrt{C_{\epsilon}^2+1} + C\right) + c\exp(-cN).
\end{align*}
Choose $\delta = \left( \frac{cC^{K+1}}{JL^2 (\sqrt{C_\epsilon^2+1}+C)^2}\right)^{\frac{1}{K+3}}$, then
\begin{align*}
\frac{1}{NJ}\EEE \left( \| \tilde X - X^* \|_F^2  \Big| X^* \right) \leq c C_{\epsilon}^{\frac{K+1}{K+3}}J^{\frac{-1}{K+3}} + c\exp(-cN),
\end{align*}
which implies
\begin{align*}
\frac{1}{NJ}\EEE \left( \| \tilde X - X^* \|_F^2  \mid \mathcal A_{N,J} \right) \leq g(N,J),
\end{align*}
where we define
$g(N,J) := c C_{\epsilon}^{\frac{K+1}{K+3}}J^{\frac{-1}{K+3}} + c\exp(-cN)$.
By Chebyshev's inequality, for any $\Delta_{N,J}>0$,
\begin{align*}
\Pr\left( \frac{1}{NJ} \| \tilde X - X^* \|_F^2 \geq \frac{g(N,J)}{\Delta_{N,J}}  \Big| \mathcal A_{N,J} \right) \leq \Delta_{N,J}.
\end{align*}
Thus,
\begin{align}\label{eq:bnj}
\Pr\left( \frac{1}{NJ} \| \tilde X - X^* \|_F^2 \leq \frac{g(N,J)}{\Delta_{N,J}}  \Big| \mathcal A_{N,J} \right) \geq 1-\Delta_{N,J}.
\end{align}
Let $\mathcal B_{N,J} := \mathcal A_{N,J} \cap \{\frac{1}{NJ} \| \tilde X - X^* \|_F^2 \leq \frac{g(N,J)}{\Delta_{N,J}}\},$ then according to \eqref{eq:bnj} for any sequence $\Delta_{N,J}$ satisfying $\Delta_{N,J}=o(1)$, we have
 $$\lim_{N,J\to\infty}\Pr\left(\mathcal B_{N,J}\right) = \lim_{N,J\to\infty}\Pr\left(\mathcal A_{N,J}\right)\cdot\lim_{N,J\to\infty}\Pr\left(\frac{1}{NJ} \| \tilde X - X^* \|_F^2 \leq \frac{g(N,J)}{\Delta_{N,J}}\Big|\mathcal A_{N,J}\right)=1.$$
We will restrict our analysis on $\mathcal B_{N,J} $ in what follows. Let $h(N,J) = \frac{g(N,J)}{\Delta_{N,J}}$, then on $\mathcal B_{N,J},$ we have $\frac{1}{NJ} \| \tilde X - X^* \|_F^2 \leq h(N,J).$

Recall $C_{\epsilon} = \frac{h(2\epsilon)}{C}$. Then, according to the definition of the function $h$ and $C_{\epsilon}$, we can see that  $f(CC_{\epsilon}),f(-CC_{\epsilon}) \in [2\epsilon, 1-2\epsilon].$ This interval is non-empty because $\epsilon \leq \frac{1}{4}.$ Thus, when the event $\mathcal B_{N,J}$ happens, we have $x^*_{ij} = f((\ttt_i^+)^\top\aaaa_j^+) \in [2\epsilon, 1-2\epsilon],$ which leads to
\begin{align*}
\frac{1}{NJ} \sum_{i,j} 1_{\{ \tilde x_{ij} \notin [\epsilon,1-\epsilon] \}}\leq \frac{1}{NJ} \sum_{i,j} 1_{\{| \tilde x_{ij}- x_{ij}^*| \geq \epsilon \}}\leq \frac{1}{NJ} \sum_{i,j} \frac{( \tilde x_{ij}- x_{ij}^*)^2}{\epsilon^2}
 \leq \frac{h(N,J)}{\epsilon^2}.
\end{align*}
Since $\hat X$ and $\tilde X$ are not far away from each other by definition, we can bound $\| \hat X - X^* \|_F^2$ by
\begin{align}
\frac{1}{NJ} \| \hat X - X^* \|_F^2 &= \frac{1}{NJ} \sum_{i,j} \left[ (\tilde x_{ij} - x_{ij}^*)^2 1_{\{ \tilde x_{ij} \in [\epsilon, 1-\epsilon] \}} + (\hat x_{ij} - x_{ij}^*)^2 1_{\{ \tilde x_{ij} \notin [\epsilon, 1-\epsilon] \}} \right] \nonumber\\
&\leq \frac{1}{NJ} \sum_{i,j} (\tilde x_{ij} - x^*_{ij})^2 + \frac{1}{NJ} \sum_{i,j} \left(1-3\epsilon\right)^2 1_{\{ \tilde x_{ij} \notin [\epsilon,1-\epsilon] \}} \nonumber\\
&\leq \left( 1+ \left(\frac{1-3\epsilon}{\epsilon}\right)^2 \right) h(N,J)\nonumber\\
&\leq \frac{1}{\epsilon^2} h(N,J)
\end{align}
where the last inequality is because $\epsilon \leq \frac{1}{4}$.
According to condition A3 and the above inequality, we have
\begin{align}\label{eq:bound2}
\frac{1}{NJ} \| \tilde M - M^* \|_F^2 &= \frac{1}{NJ} \| f^{-1}(\hat X) - f^{-1}(X^*) \|_F^2 \\
&\leq{\frac{1}{(g(\epsilon))^2}\frac{1}{NJ} \| \hat X - X^* \|_F^2} \\
&\leq \frac{1}{(\epsilon g(\epsilon))^2} h(N,J).
\end{align}
The first inequality holds because $x_{ij}^*, \hat x_{ij} \in [\epsilon,1-\epsilon]$ on the event $\mathcal B_{N,J}$.

We proceed to an upper bound of $\hat M - \Theta^*(A^*)^\top$.
Recall that $M^* = \E_N (d^*)^\top + \Theta^*(A^*)^\top, \tilde M = \hat M + \E_N \hat d.$ Let $H_1 = \hat M - \Theta^*(A^*)^\top$ and $H_2 = \E_N(\hat d)^\top - \E_N (d^*)^\top.$ We have
\begin{align}\label{eq:eqforbounds}
\frac{1}{NJ} \| H_1 + H_2 \|_F^2 = \frac{1}{NJ} \left(  \|H_1\|_F^2 + \|H_2\|_F^2 + 2tr\{H_1^\top H_2\} \right).
\end{align}
We first bound the trace term in the above display,
\begin{align*}
\left|tr\{H_1^\top H_2\}\right| &= \left| tr\{ (A^*(\Theta^*)^\top - \hat M^\top)\E_N(\hat \dd - \dd^*)^\top \}  \right| \\
&=   \left| tr\{ A^*(\Theta^*)^\top\E_N(\hat \dd - \dd^*)^\top \}  \right|,  \quad (\hat M^\top\E_{N} = \0_J)\\
&= \left| (\hat \dd - \dd^*)^\top A^* (\Theta^*)^\top \E_N  \right|, \quad \text{(exchangeability for trace operator)}\\
&= \left| \left\langle \sum_j (\hat d_j - d_j^*)\aaaa_j^*, \sum_i \ttt_i^*  \right\rangle  \right| \\
&\leq \left\| \sum_j (\hat d_j - d^*_j)\aaaa_j^* \right\| \left\| \sum_i \ttt_i^* \right\|. \quad \text{(Cauchy-Schwarz inequality)}
\end{align*}
Through simple algebra, we have $d_j^* = \frac{1}{N} \sum_{i=1}^N \left( m^*_{ij} + (\ttt_i^*)^\top\aaaa_j^*\right).$ By the definition of $\hat d_j,$ we have $\hat d_j = \frac{1}{N}\sum_{i=1}^N \tilde m_{ij}.$ Then
\begin{align*}
|\hat d_j - d^*_j| &\leq \left| \frac{1}{N} \sum_i (\tilde m_{ij} - m_{ij}^*) \right| + \left| \frac{1}{N} \sum_i (\ttt_i^*)^\top \aaaa_j^* \right| \\
&\leq \left| \frac{1}{N} \sum_i(\tilde m_{ij} - m_{ij}^*) \right| + \left\| \frac{1}{N}\sum_i\ttt_i^* \right\| \left\| \aaaa_j^* \right\|,
\end{align*}
which leads to
\begin{align*}
\left\| \sum_j (\hat d_j - d^*_j)\aaaa_j^* \right\| &\leq \sum_j |\hat d_j - d_j^*| \| \aaaa_j^*\| \\
&\leq C\sum_j |\hat d_j - d_j^*|, \quad (\|\aaaa_j^*\| \leq C)\\
&\leq C\sum_j \left\{ \left| \frac{1}{N}\sum_i (\tilde m_{ij} - m_{ij}^*) \right| + \left\| \frac{1}{N}\sum_i\ttt_i^* \right\| \left\| \aaaa_j^* \right\|  \right\}\\
&\leq \frac{C}{N}\sum_{i,j} |\tilde m_{ij} - m_{ij}^*| + C^2 J\left\| \frac{1}{N}\sum_i \ttt_i^* \right\|, \quad (\|\aaaa_j^*\| \leq C)\\
&\leq CJ\sqrt{ \frac{1}{NJ}\| \tilde M - M^* \|_F^2 } + C^2J \left\| \frac{1}{N}\sum_i\ttt_i^* \right\|. \quad \text{(Cauchy-Schwarz inequality)}
\end{align*}
So we can bound $\left|tr\{H_1^\top H_2\}\right|$ by
\begin{align}\label{eq:bound1}
\left|tr\{H_1^\top H_2\}\right| \leq \left(CJ\sqrt{ \frac{1}{NJ}\| \tilde M - M^* \|_F^2 } + C^2J \left\| \frac{1}{N}\sum_i\ttt_i^* \right\| \right) \left\| \sum_i \ttt_i^* \right\|
\end{align}
According to condition A2 and law of large number, we have
\begin{align*}
\Pr\left(  \frac{1}{N} \left\|\sum_{i=1}^N \ttt_{i}^* \right\| \leq \xi  \right) \to 1, \quad \text{as } N,J \to \infty,
\end{align*}
for any $\xi>0.$
Let
\begin{align*}
\mathcal C_{N,J,\xi} :=  \left\{  \frac{1}{N} \left\|\sum_{i=1}^N \ttt_{i}^* \right\| \leq \xi  \right\}\cap \mathcal B_{N,J},
\end{align*}
then we have
\begin{align*}
\Pr( \mathcal C_{N,J,\xi} ) \to 1, \quad \text{as } N,J \to \infty,
\end{align*}
for any $\xi>0.$
On $\mathcal C_{N,J,\xi},$ according to \eqref{eq:bound2} , \eqref{eq:eqforbounds} and \eqref{eq:bound1},
\begin{align}
\frac{1}{NJ} \| \Theta^*(A^*)^\top - \hat M  \|_F^2 = \frac{1}{NJ} \| H_1 \|_F^2 &\leq \frac{1}{NJ} \| \tilde M - M^* \|_F^2 + \frac{2}{NJ} \left|tr\{H_1^\top H_2\}\right| \nonumber\\
&\leq \frac{h(N,J)}{(\epsilon g(\epsilon))^2} + C\xi \left( \frac{\sqrt{h(N,J)}}{\epsilon g(\epsilon)} + C\xi \right).\label{eq:ineq1}
\end{align}
Recall how we get $\hat \Theta, \hat A$ in algorithm \ref{alg:SVD2} and we have
\begin{align*}
&\| \hat M - \hat \Theta \hat A^\top \|_2 \\ =&  \sigma_{K+1}(\hat M)\\ =& | \sigma_{K+1}(\hat M) - \sigma_{K+1}(\Theta^*(A^*)^\top) |\nonumber\\
\leq & \| \hat M - \Theta^*(A^*)^\top \|_2\nonumber\\
 \leq &\| \Theta^*(A^*)^\top - \hat M \|_F.
\end{align*}
So
\begin{align}\label{eq:ineq2}
\| \hat\Theta\hat A^\top - \Theta^*(A^*)^\top \|_2 \leq \| \hat\Theta\hat A^\top - \hat M \|_2 + \| \hat M - \Theta^*(A^*)^\top \|_2 \leq 2\| \hat M - \Theta^*(A^*)^\top \|_F,
\end{align}
which leads to
\begin{align}
\frac{1}{NJ} \| \hat\Theta\hat A^\top - \Theta^*(A^*)^\top \|_F^2 &\leq \frac{2K}{NJ} \| \hat\Theta\hat A^\top - \Theta^*(A^*)^\top \|_2^2, \quad   \nonumber\\
&\leq \frac{8K}{NJ} \| \hat M - \Theta^*(A^*)^\top\|_F^2,  \nonumber\\
&\leq  8K \frac{h(N,J)}{(\epsilon g(\epsilon))^2} + 8KC\xi \left( \frac{\sqrt{h(N,J)}}{\epsilon g(\epsilon)} + C\xi \right), \label{eq:rate2}
\end{align}
where the first inequality is due to $\textrm{rank}\big( \hat\Theta\hat A^\top - \Theta^*(A^*)^\top \big) \leq 2K $, the second inequality is due to \eqref{eq:ineq2} and the last inequality is due to \eqref{eq:ineq1}. Thus, on the event $\mathcal C_{N,J,\xi}$
\begin{equation*}
	\frac{1}{NJ} \| \hat\Theta\hat A^\top - \Theta^*(A^*)^\top \|_F^2 = O\left(
\frac{h(N,J)}{(\epsilon g(\epsilon))^2} + \xi \left( \frac{\sqrt{h(N,J)}}{\epsilon g(\epsilon)} + \xi \right)
	\right).
\end{equation*}

Recall
\begin{align*}
\frac{h(N,J)}{(\epsilon g(\epsilon))^2} = \frac{c}{\Delta_{N,J}}\left( \frac{(h(2\epsilon))^{\frac{K+1}{K+3}}}{(\epsilon g(\epsilon))^2 J^{\frac{1}{K+3}}} + \frac{\exp(-cN)}{(\epsilon g(\epsilon))^2}\right),
\end{align*}
where $\Delta_{N,J}$ could be any sequence satisfying $\Delta_{N,J}=o(1)$.
By \eqref{eq:epsilon1}, \eqref{eq:epsilon2} and condition A5, there exists $\Delta_{N,J} = o(1)$ such that $\frac{h(N,J)}{(\epsilon g(\epsilon))^2} = o(1)$.
So fix any $\xi < 1,$ for $N,J$ large enough, we have $\frac{h(N,J)}{(\epsilon g(\epsilon))^2} \leq \xi.$
{Then there is a constant $\kappa$ such that for $N,J$ large enough, on $C_{N,J,\xi}$ with $\xi\in (0,1)$, we have,
\begin{equation}
	\frac{1}{NJ} \| \hat\Theta\hat A^\top - \Theta^*(A^*)^\top \|_F^2 \leq \kappa \xi.
\end{equation}
This combined with $\Pr(C_{N,J,\xi})\to 1$ for any $\xi$ sufficiently small completes the proof.}
\end{proof}

\begin{proof}[Proof of Lemma \ref{lem:A consis}]
Let $$Q^{(N,J)} = \frac{1}{\sqrt{N}} \hat\Theta^\top\Theta^*\left( (\Theta^*)^\top\Theta^* \right)^{-\frac{1}{2}}$$ and in the following we will show that
\begin{align*}
\frac{1}{JK}\| A^* - \hat AQ^{(N,J)} \|^2_F \overset{pr}{\rightarrow} 0.
\end{align*}

For any $\alpha > 0,$ let
\begin{align}\label{eq:event D}
\mathcal D_{N,J,\alpha} := \left\{1-\alpha \leq \frac{\sigma_K(\Theta^*)}{\sqrt{N}} \leq \frac{\sigma_1(\Theta^*)}{\sqrt{N}} \leq 1+\alpha \right\}.
\end{align}
Applying Theorem 5.39 of \cite{vershynin2010introduction} to the matrix $\Theta^*$,  we have $\lim_{N,J\to\infty}\Pr(\mathcal D_{N,J,\alpha}) = 1$ for any $\alpha>0$.
We restrict our analysis on $\mathcal D_{N,J,\alpha}$ in what follows and denote $$Q(N,J):= \frac{1}{NJ}\|\hat \Theta\hat A^\top - \Theta^*(A^*)^\top \|_F^2.$$
Then,
\begin{equation}\label{eq:A bound}
\begin{aligned}
\| A^* - \hat AQ^{(N,J)} \|_F =& \| A^* - \hat A \frac{1}{\sqrt{N}}\hat\Theta^\top\Theta^*\left( (\Theta^*)^\top\Theta^* \right)^{-\frac{1}{2}} \|_F \\
\leq& \underbrace{\| A^* - A^* \frac{1}{\sqrt{N}}\left( (\Theta^*)^\top\Theta^* \right)^{\frac{1}{2}} \|_F}_{(a)} + \underbrace{\| (A^*(\Theta^*)^\top - \hat A \hat\Theta^\top ) \frac{1}{\sqrt{N}} \Theta^*((\Theta^*)^\top\Theta^*)^{-\frac{1}{2}} \|_F}_{(b)}.
\end{aligned}
\end{equation}

We consider $(b)$ first:
\begin{equation}
\begin{aligned}
(b) &\leq \| A^*(\Theta^*)^\top - \hat A\hat \Theta^\top \|_F \frac{1}{\sqrt{N}} \| \Theta^* \|_2 \| ((\Theta^*)^\top\Theta^*)^{-\frac{1}{2}} \|_2\\
&= \sqrt{NJ} \sqrt{Q(N,J)}   \frac{\sigma_1(\Theta^*)}{\sqrt{N}}  \frac{1}{\sigma_K(\Theta^*)}\\
&\leq \sqrt{JQ(N,J)}\frac{1 + \alpha}{1-\alpha},\quad (\text{by } \eqref{eq:event D})\label{eq:(b)}
\end{aligned}
\end{equation}

For (a), notice that
\begin{align*}
\left\| \frac{1}{\sqrt{N}} \left( (\Theta^*)^\top\Theta^* \right)^{\frac{1}{2}} - I_K \right\|_2 &= \max_{1\leq k\leq K} \left|\frac{\sigma_k(\Theta^*)}{\sqrt{N}} -1 \right|  \\
&\leq \alpha. \quad (\text{by } \eqref{eq:event D})
\end{align*}
So
\begin{align}\label{eq:(a)}
(a) \leq \|A^*\|_F \left\| \frac{1}{\sqrt{N}} \left( (\Theta^*)^\top\Theta^* \right)^{\frac{1}{2}} - I_K \right\|_2 \leq C\sqrt{J}\alpha.
\end{align}

Combine \eqref{eq:A bound}, \eqref{eq:(b)} and \eqref{eq:(a)}, we get on $\mathcal D_{N,J,\alpha}$
\begin{align*}
\frac{1}{\sqrt{JK}}\| A^* - \hat AQ \|_F \leq& \frac{C\alpha}{\sqrt{K}} +\frac{1+\alpha}{\sqrt{K}(1-\alpha)}\sqrt{Q(N,J)}.
\end{align*}
Recall that $Q(N,J) = \frac{1}{NJ}\|\hat \Theta\hat A^\top - \Theta^*(A^*)^\top \|_F^2 \overset{pr}{\rightarrow} 0, \alpha$ can be arbitrarily small and $\Pr(\mathcal D_{N,J,\alpha}) \to 1,$ we complete the proof.
\end{proof}

\begin{proof}[Proof of Lemma \ref{lem:chatterjee}]
This lemma is almost the same as Theorem 1.1 of \cite{chatterjee2015matrix} by setting, in his notations, $\eta = 0.02$ and $\sigma^2 = 1/4,$
except two small differences. The first is that the probability $p$ can be changed through $N,J$ in the setting of \cite{chatterjee2015matrix} while $p$ is a constant in our setting. Therefore we absorb $p$ into constants $c$ in the LHS of \eqref{eq:chatterjee}. The second difference is a modification in step 5 of Algorithm \ref{alg:SVD2} that we require $X$ to include at least $K+1$ singular values of $Z.$
This does not change the result of Theorem 1.1 of \cite{chatterjee2015matrix} given the following lemma which is based on Lemma 3.5 of \cite{chatterjee2015matrix}.

\begin{lemma}\label{lem:chatterjee 3}
For fixed $0 < m \leq n$ and a $m\times n$ matrix $A$,
let $A = \sum_{i=1}^m \sigma_i x_iy_i^\top$ be the singular value decomposition of A. Fix any $\delta > 0$ and integer $T>0$, and define $$\tilde B := \sum_{i=1}^{l} \sigma_i x_iy_i^\top,$$ where $l = \max\{T, \argmax\{i: \sigma_i > (1+\delta) \|A - B\|\}\}.$ Then
\begin{equation}\label{eq:bound B}
\|\tilde B - B\|_F \leq (1+\delta) \sqrt{T} \|A-B\| + K(\delta) \left( \|A-B\| \|B\|_* \right)^\frac{1}{2},
\end{equation}
where $K(\delta) = (4+2\delta) \sqrt{2/\delta} + \sqrt{2+\delta}.$
\end{lemma}
Notice that we have another term $(1+\delta) \sqrt{T} \|A-B\|$ in \eqref{eq:bound B} compared with Lemma 3.5 in \cite{chatterjee2015matrix}, which is due to the composition of $\tilde B.$
In the proof of Theorem 1.1 in \cite{chatterjee2015matrix}, by replacing Lemma 3.5 in \cite{chatterjee2015matrix} by the above lemma with $T = K+1$, we get
\begin{equation}\label{eq:1/J term}
\frac{1}{NJ}\EEE \left( \| \tilde X - X^* \|_F^2  \Big| X^* \right) \leq c \min \left\{ \frac{\|X^*\|_*}{J\sqrt{N}} + \frac{1}{J}, \frac{\|X^*\|^2_*}{NJ}, 1 \right\}  + c e^{-cN}.
\end{equation}
The $1/J$ term in \eqref{eq:1/J term} results from the first term in \eqref{eq:bound B}.
Notice that if $$\frac{\|X^*\|_*}{J\sqrt{N}} + \frac{1}{J} \leq \frac{\|X^*\|^2_*}{NJ},$$ then $$\frac{\|X^*\|_*}{J\sqrt{N}} \leq \frac{\|X^*\|^2_*}{NJ},$$ which leads to $$\frac{\|X^*\|_*}{J\sqrt{N}} \geq \frac{1}{J}.$$ Therefore we can remove the $1/J$ term in \eqref{eq:1/J term} to complete the proof.
\end{proof}

\begin{proof}[Proof of Lemma \ref{lem:chatterjee 3}]
Let $$\hat B := \sum_{i: \sigma_i > (1+\delta)\|A-B\|}\sigma_ix_iy_i^\top$$ and by Lemma 3.5 of \cite{chatterjee2015matrix}, we have $$\|\tilde B - B\|_F \leq  K(\delta) \left( \|A-B\| \|B\|_* \right)^\frac{1}{2}.$$ Note that $$ \| \tilde B - \hat B \|_F \leq \sqrt{T}(1+\delta)\|A-B\|$$ and we complete the proof by triangular inequality.
\end{proof}

\vspace{3cm}

\bibliographystyle{apa}
\bibliography{ref}

\end{document}